%% file: paper.tex
\documentclass[11pt]{article}
\usepackage{times}
\usepackage{amsmath,amsfonts,amssymb}
\usepackage{epsfig,amstext,xspace}
\usepackage{amsthm}
\usepackage{tabu}
\usepackage{multirow}
\usepackage{booktabs}
\usepackage{color}

\usepackage[noend]{algpseudocode}
\usepackage{algorithm}
\usepackage{caption}
\usepackage{subcaption}
\usepackage{authblk}

\newcommand{\Pij}{ { {\tt Pj}_{\sigma^{T^i,i}_\ell } }}

\newtheorem{theorem}{Theorem}[section]
\newtheorem{fact}[theorem]{Fact}
\newtheorem{lemma}[theorem]{Lemma}
\newtheorem{corollary}[theorem]{Corollary}
\newtheorem{claim}[theorem]{Claim}

\newtheorem{definition}[theorem]{Definition}

{\theoremstyle{remark} \newtheorem{remark}[theorem]{Remark}}
{\theoremstyle{definition}

}

\newcommand{\bg}[1]{\medskip\noindent{\bf #1}}

\def\QQ{\mathbb Q}
\def\RR{\mathbb R}
\def\ZZ{\mathbb Z}

\def\cB{\mathcal B}
\def\cC{\mathcal C}

\def\cG{\mathcal G}

\def\cH{\mathcal H}

\def\cL{\mathcal L}

\def\cT{\mathcal T}

\def\cU{\mathcal U}

\def\cX{\mathcal X}

\def\bzero{\mathbf 0}

\newcommand{\hide}[1]{}
\newcommand{\raf}[1]{(\ref{#1})}

\newcommand{\cP}{\ensuremath{\mathcal{P}}}
\newcommand{\cQ}{\ensuremath{\mathcal{Q}}}
\newcommand{\cR}{\ensuremath{\mathcal{R}}}
\newcommand{\cV}{\ensuremath{\mathcal{V}}}

\newcommand{\Pc}{\ensuremath{\mathcal P}}

\newcommand{\U}{\ensuremath{\mathcal U}}

\newcommand{\OPT}{\ensuremath{\textsc{Opt}}}

\newcommand{\rank}{\ensuremath{\mathrm{rank}}}
\newcommand{\cpr}{\ensuremath{\mathrm{cp\text{-}rank}}}

\newcommand{\conv}{\operatorname{conv}}
\newcommand{\poly}{\operatorname{poly}}
\newcommand{\polylog}{\operatorname{polylog}}

\newcommand{\bone}{\ensuremath{\boldsymbol{1}}}

\newcommand{\opt}{\ensuremath{\mathsf{opt}}}

\newcommand{\argmax}{\operatorname{argmax}}

\def\PQC{{\sc Packing-BQC}}
\def\PL{{\sc Pack-Lin}}
\def\PS{{\sc Pack-Sub}}
\def\PQ{{\sc Pack-Quad}}
\def\PLPTAS{{\sc Pack-Lin-PTAS}}
\def\PSAPPROX{{\sc Pack-Sub-Approx}}
\def\CQC{{\sc Covering-BQC}}
\def\CL{{\sc Cover-Lin}}
\def\CS{{\sc Cover-Sub}}
\def\CQ{{\sc Cover-Quad-QC}}
\def\GRC{{\sc Greedy-Cover}}
\def\CLQPTAS{{\sc Cover-Lin-QPTAS}}

\newenvironment{changemargin}[2]{%
\begin{list}{}{%
\setlength{\topsep}{0pt}%
\setlength{\leftmargin}{#1}%
\setlength{\rightmargin}{#2}%
\setlength{\listparindent}{\parindent}%
\setlength{\itemindent}{\parindent}%
\setlength{\parsep}{\parskip}%
}%
\item[]}{\end{list}}

\newcommand{\fullonly}[1]{#1}
\newcommand{\sodaonly}[1]{}

\title{Approximation Schemes for Binary Quadratic Programming Problems with Low cp-Rank Decompositions}  

\author{Khaled Elbassioni\thanks{Department of Electrical Engineering and Computer Science, Masdar Institute of Science and Technology, Abu Dhabi, UAE}
\qquad Trung Thanh Nguyen\thanks{Department of Electrical Engineering and Computer Science, Masdar Institute of Science and Technology, Abu Dhabi, UAE}}
\affil{}

\begin{document}

\maketitle
\begin{abstract}
Binary quadratic programming problems have attracted much attention in the last few decades due to their potential applications. This type of problems are NP-hard in general, and still considered a challenge in the design of efficient approximation algorithms for their solutions. The purpose of this paper is to investigate the approximability for a class of such problems where the constraint matrices are {\it completely positive} and have low {\it cp-rank}. In the first part of the paper, we show that a completely positive rational factorization of such matrices can be computed in polynomial time, within any desired accuracy. 
We next consider binary quadratic programming problems of the following form: Given matrices $Q_1,\ldots,Q_n\in\RR_+^{n\times n}$, and a system of $m$ constrains $x^TQ_ix\le C_i^2$ ($x^TQ_ix\ge C_i^2$), $i=1,\ldots,m$, we seek to find a vector $x^*\in \{0,1\}^n$ that maximizes (minimizes) a given function $f$. This class of problems generalizes many fundamental problems in discrete optimization such as packing and covering integer programs/knapsack problems, quadratic knapsack problems, submodular maximization, etc. We consider the case when $m$ and the cp-ranks of the matrices $Q_i$ are bounded by a constant.

Our approximation results for the maximization problem are as follows. For the case when the objective function is nonnegative submodular, we give an $(1/4-\epsilon)$-approximation algorithm, for any $\epsilon>0$; 
when the function $f$ is linear, we present a PTAS. We next extend our PTAS result to a wider class of non-linear objective functions including quadratic functions, multiplicative functions, and sum-of-ratio functions. 
The minimization problem seems to be much harder due to the fact that the relaxation is {\it  not} convex. For this case, we give a QPTAS for $m=1$.
\end{abstract}

\newpage

\pagenumbering{arabic}
\normalsize

\input{intro}

\input{cp-rank}
\input{approx-scheme}

\input{minimization}
\input{application}


\section{Conclusion}\label{conc}
In this paper, we have investigated the approximability of binary quadratic programming problems when the cp-rank of the constraints' matrices are bounded by a constant. It turns out that bounding the cp-rank of these matrices makes several interesting variants of the quadratic programming problem tractable. In particular, our results hint that limiting the cp-rank makes the quadratic problem exhibit similar approximability behavior as the linear case, assuming a constant number of packing or covering quadratic constrains.  For the case with any number of quadratic constraints and linear objective function, one can easily obtain a $O(m\sqrt{r})$-approximation algorithm for the packing problem
, where $r=\max_{i\in [m]}r_i$. The first interesting question is whether there exists a (Q)PTAS for the covering problem with a constant number of quadratic constraints.  
Extending our results to the case of low nonnegative ranks, and finding non-trivial inapproximability bounds, parametrized by the cp-rank of the constraints matrices are also another interesting open problems.

\bibliographystyle{alpha} 
\bibliography{reference}

\end{document}

%% file: intro.tex
\section{Introduction}
Binary quadratic programming is a classical class of combinatorial optimization problems containing a quadratic function either in its objective or in its constraints.
The study of this problem has been of great theoretical interest in the last few decades due to its wide range of applications, for example, capital budgeting and financial analysis \cite{lau:j:quadratic-binary-programming-with-application-to-capital-budgeting-problem,mc-yor:j:an-implicit-enumeration-algorithm-for-quadratic-integer-programming}, machine scheduling \cite{ali-koc-ahm:j:0-1-quadratic-programming-approach-for-the-optimal-solution-of-two-scheduling-problems}, and traffic message management problems \cite{gal-ham-sim:j:quadratic-knapsack-problems}.
One can mainly distinguish two types of special classes of binary quadratic programming that have been studied extensively in the literature. The first class involves maximizing a quadratic function without any structural
constraints (a.k.a {\em unconstrained $0$-$1$ quadratic problem}) which is is known to be APX-hard since it generalizes, for example, the maximum cut problem. The second class deals with the {\em quadratic knapsack problem} (QKP) arising from the classical knapsack problem by maximizing a quadratic objective function subject to (linear non-negative) knapsack constraint(s) (see. e.g., \cite{pis:j:the-quadratic-knapsack-problem-a-survey}). Very recently, Yu and Chau \cite{yu-cha:c:complex-knapsack-problem-and-incentive-in-AC-power-systems} introduced a new variant of the knapsack problem,  called {\em the complex-demand knapsack problem} (CKP), that arose from the allocation of power in AC (alternating current) electrical systems. In this setting, the constraint is exactly a sum of two squares of linear non-negative functions, and thus is a generalization of the classical knapsack constraint.  

Binary programming with quadratic objective function and/or quadratic constraints is hard to approximate in general. For example, quadratic knapsack problem is NP-hard to approximate to within any finite worst case factor \cite{rad-woe:j:the-quadratic-0-1-knapsack-problem-with-series-parallel-support}. Most algorithms for this kind of
problems focus on handling various special cases of the general problem
 \cite{rad-woe:j:the-quadratic-0-1-knapsack-problem-with-series-parallel-support,pfe-sch:c:approximating-the-quadratic-knapsack-problem-on-special-graph-classes,kel-str:j:fully-polynomial-approximation-schemes-for-a-symmetric-quadratic-knapsack,kel-str:j:the-symmetric-quadratic-problem-approximation-and-scheduling-applications,xu:j:a-strongly-polynomial-fptas-for-the-symmetric-quadratic-knapsack-problem}. Those papers investigate the existence of approximation schemes by exploiting the graph-theoretic structure of QKPs or taking into account the special multiplicative structure of the coefficients of the objective function. Another remarkable direction in studying the approximability of binary quadratic programming is to restrict the objective function to certain special types, in particular, those having {\em low-rank}\fullonly{\footnote{A function $f : \RR^n \rightarrow \RR$ is said to be of rank $k$, if there exists a continuous function $g : \RR^k \rightarrow  \RR $ and linearly independent vectors $a_1,\ldots , a_k \in \RR^n$ such that $f (x) = g(a^T_1 x, \ldots , a^T_k x)$, (see \cite{kel-nik:c:on-the-hardness-and-smoothed-complexity-of-quasi-concave-minimization}).}} \cite{kel-nik:c:on-the-hardness-and-smoothed-complexity-of-quasi-concave-minimization}. Intuitively, a function of low-rank can be expressed as a combination of a (fixed) number of linear functions. Allemand et al. \cite{all-fuk-lie-ste:j:a-polynomial-case-unconstrained-quadratic-optimization}
study the unconstrained 0-1 quadratic maximization problem and propose a polynomial-time algorithm when the symmetric matrix describing the coefficients of the quadratic objective function is positive semidefinite and has fixed rank\fullonly{\footnote{Note that in this case the objective function can be written as a sum of the squares of not more than $k$ linear functions, where $k$ is the rank of the coefficient matrix.}}. Kern and Woeginger \cite{ker-woe:j:quadratic-programming-and-combinatorial-minimum-weight-product-problems}
and Goyal et al. \cite{goy-kay-rav:j:an-fptas-for-minimizing-product-two-non-negative-linear-cost-function} give an FPTAS for minimizing the product of two non-negative linear
functions when the convex hull of feasible integer solutions is given in terms of linear inequalities. Mittal and Schulz \cite{mit-sch:j:an-fptas-optimizing-a-class-of-lower-rank-functions-over-a-polytope,mit-sch:j:a-general-framework-for-designing-approximation-scheme} and Goyal and Ravi \cite{goy-rav:j:an-fptas-for-minimizing-class-of-low-rank-quasi-concave-function} extend this result to the more general class of low-rank functions.

In this paper, we consider a class of 0-1 quadratic programming, where a (nonnegative) linear or submodular objective function has to be minimized or maximized subject to a system of quadratic inequalities involving nonnegative semidefinite matrices with low {\em completely positive ranks}. Intuitively, the completely positive rank (cp-rank) of a semidefinite matrix $Q$ is the smallest positive integer $r$ such that $Q$ can be decomposed into a sum of $r$ non-negative rank-one matrices. As a consequence, every quadratic constraint involving such matrices can be written as the sum of $r$ squares of (non-negative) linear functions which we call {\em cp-decomposition}. Hence, our problem can be seen as a generalization of the submodular maximization with (linear) knapsack constraints which has been studied extensively in the literature \cite{svi:c:note-on-maximizing-submodular-function-knapsack-constraint,kul-sha-tam:c:maximizing-submodular-set-function-subject-to-multiple-linear-constraint,lee-mir-nag-svi:c:non-monotone-submodular-maximization-under-matroid-and-knapsack-constraints}. Herein, we would like to emphasize that the cp-decomposition plays an important role not only in binary quadratic programming \cite{bur:j:on-the-copositive-representation-of-binary-and-continuous-non-convex-quadratic-programs}, but also in in many other areas such as block designs in combinatorial analysis \cite{hall:b:combinatorial-theory}, data mining \cite{zas-sha:c:a-unifying-approach-to-hard-and-probabilistic-clustering,he-xie-zdu-zho-cic:j:symmetric-nonnegative-matrix-factorization-application-probabilistic-clustering} and economics \cite{van-ho-van-doo:j:on-the-parameterization-of-the-creditrisk+-model}.
Several attempts have been devoted in the last decades to settling the complexity of determining if a given semidefinite positive matrix has finite cp-rank.
Surprisingly, however, this has been settled only very recently in \cite{dic-gij:j:on-the-computational-complexity-of-membership-problem-for-cp-matrix-cone}, where it was shown that this problem is NP-hard. In this paper we restrict out attention to the case of low cp-rank and study the polynomial-time algorithm for finding a cp-decomposition of a matrix whose cp-rank is fixed.
 
\subsection{Our Contribution.}
We summarize our contributions in this paper as follows:
\begin{itemize}
	\item We give an $n^{O(r^2)}$ deterministic algorithm that given a positive semidefinite matrix $Q$ of size $n$ decides if a nonnegative factorization of inner dimension at most $r$ exists. Furthermore, a rational approximation to the solution to within an additive accuracy of $\delta$ can be computed in time $\text{poly}(\mathcal{L},n^{O(r^2)},\log\frac{1}{\delta})$, where $\cL$ is the bit length of the input. 
	\item We present a PTAS for maximizing a linear function subject to a fixed number of quadratic packing constraints involving positive semidefinite matrices with low cp-ranks. We provide an $(1-\epsilon)(\frac{1}{4} - \epsilon)$-approximation algorithm for the problem of maximizing a submodular function, for any $\epsilon>0$. For the case when the submodular function is monotone, the approximation factor is $(1-\epsilon)(1-\frac{1}{e}-\epsilon)$, for any $\epsilon > 0$. Our results generalize the ones of Sviridenko \cite{svi:c:note-on-maximizing-submodular-function-knapsack-constraint} and of Kulik et al. \cite{kul-sha-tam:c:maximizing-submodular-set-function-subject-to-multiple-linear-constraint}.
	\item We extend our PTAS result to a wider class of non-linear objective functions including quadratic functions, multiplicative functions, and sum-of-ratio functions.
      \item We give a quasi-PTAS for minimizing a linear function subject to one quadratic covering constraint involving a positive semidefinite matrix with low cp-rank.
\end{itemize}

\subsection{Previous Work.}

The problem of maximizing a nonnegative submodular function under knapsack constraints has been studied extensively in the literature. Sviridenko \cite{svi:c:note-on-maximizing-submodular-function-knapsack-constraint} studied the problem for monotone submodular functions and gave a $(1 - 1/e)$-approximation for the case of one knapsack constraint. Kulik et al. \cite{kul-sha-tam:c:maximizing-submodular-set-function-subject-to-multiple-linear-constraint} considered the problem with any constant number of knapsack constraints, and gave a $(1 - 1/e -\epsilon)$-approximation, for any $\epsilon>0$. Lee et al. \cite{lee-mir-nag-svi:c:non-monotone-submodular-maximization-under-matroid-and-knapsack-constraints} investigated the problem of maximizing a general submodular function under a constant number of knapsack constraints, and presented algorithms that achieve approximation ratio of $1/5-\epsilon$. A better approximation factor of $1/4-\epsilon$ was obtained in \cite{kul-sha-tam:c:maximizing-submodular-set-function-subject-to-multiple-linear-constraint}. In this paper, we will extend these results to the case with {\em quadratic} knapsack constraints having low cp-rank.

Chau et al. \cite{CEK14} considered the problem CKP with linear objective function and gave a PTAS
, which settles the complexity of the problem given the strong NP-hardness in \cite{yu-cha:c:complex-knapsack-problem-and-incentive-in-AC-power-systems}. As CKP is a special case of the problem of maximizing a linear function subject to a single quadratic inequality whose constraint matrix has cp-rank 2, our second result above can be thought of as a generalization of the  PTAS in \cite{CEK14} in several directions: we allow the cp-rank to be any constant, we consider any constant number of inequalities, and consider more general objective functions. In fact, our result for submodular functions is based essentially on the same geometric idea used in \cite{CEK14}.

%% file: cp-rank.tex
\section{Completely Positive Rank Factorizations}
\label{sec:cp-rank}

A positive semi-definite matrix $Q\succeq 0$ is said to have {\it completely positive rank} $r$, denoted $\cpr(Q)=r$, if $r$ is the smallest number of non-negative rank-one matrices 
into which the matrix $Q$ can be decomposed additively, 
i.e., $Q=\sum_{j=1}^{r}q^{j}(q^{j})^T$, where $q^{j}\in\RR^n_+$ are non-negative vectors. If no such $r$ exists then $\cpr(Q)=+\infty$. A notion related to the cp-rank is the {\em nonnegative rank} of a nonnegative (not necessarily symmetric) matrix: a positive integer $r$ is called the nonnegative rank of a matrix $P\in\mathbb{R}_+^{n\times m}$ if and only if it is the smallest number such that there exist two matrices $U\in \mathbb{R}_+^{n\times r},V\in\mathbb{R}_+^{r\times m}$ such that $P=UV$. It is worth noting that, in general, the nonnegative rank and the cp-rank can be different from each other. In fact, for a completely positive matrix $Q$, it is known that the nonnegative rank is always less than or equal to the cp-rank. We refer to the textbook \cite{ber-sha:b:completely-positive-matrices} and the references therein for more discussion about the cp-rank. 
The complexity of deciding the existence of such a cp-decomposition for a given matrix remained open for a long time. Very recently, Dickinson and Gijben \cite{dic-gij:j:on-the-computational-complexity-of-membership-problem-for-cp-matrix-cone} proved that this problem is NP-hard but left open the question whether or not it belongs to NP. The similar hardness result was also proved for the case of nonengative rank by Vavasis \cite{vav:j:on-the-complexity-of-nonnegative-matrix-factorization}. Therefore, it is natural to pay attention to special cases in which the factorization problem can be solved efficiently. Arora et al. \cite{aro-ge-kan-moi:c:computing-non-negative-rank-matrix} and Moitra \cite{moi:c:an-almost-optimal-algorithm-for-computing-nonnegative-rank} proposed exact polynomial-time algorithms for computing a nonnegative factorization when the nonnegative rank of the input matrix is constant. Their basic idea is to transform the factorization problem into a problem of finding a solution of a system of polynomial equations which is known to be solved by an algorithm (e.g, \cite{R92,BPR96}) whose running-time is polynomial in the number of equations but is exponential in the number of variables. Based on this idea, Arora et al. \cite{aro-ge-kan-moi:c:computing-non-negative-rank-matrix} gave a transformation to a system
with $nm$ equations and $2r^22^r$ variables. As a result, checking if the nonnegative rank of $P$ is at most $r$ can be done in $O((nm)^{2r^22^r})$ time. This has then been improved by Moitra \cite{moi:c:an-almost-optimal-algorithm-for-computing-nonnegative-rank} who gave a method to exponentially reduce the number of variables to $2r^2$, thus yielding a faster algorithm with running-time $(nm)^{O(r^2)}$. It was also shown by Arora et al. \cite{aro-ge-kan-moi:c:computing-non-negative-rank-matrix} that achieving an exact $(nm)^{o(r)}$-algorithm is impossible under the Exponential Time Hypothesis \cite{imp-pat:j:on-the-complexity-k-sat}. In the following, we will show that one can obtain an exact algorithm for the cp-rank case by employing the same idea as for the nonnegative rank. Assume $Q$ is rational and let $\cL$ be the maximum bit length of any element in $Q$. The result is stated in Theorem~\ref{cp-rank} below.

\begin{theorem}
\label{cp-rank}
There is an $n^{O(r^2)}$ time deterministic algorithm that given a positive semi-definite matrix $Q$ of size $n$ produces a nonnegative factorization $Q=UU^T$ of inner dimension at most $r$ if such a factorization exists. Furthermore, 
a rational approximation to the solution to within an additive accuracy of $\delta$ can be computed in time $\poly(\cL,n^{O(r^2)},\log\frac{1}{\delta})$. 
\end{theorem}
\begin{proof}
We proceed essentially along the same lines as in the proof of Theorem 2.1 in \cite{aro-ge-kan-moi:c:computing-non-negative-rank-matrix} (for finding a simplicial factorization).  
Let $Q$ be a nonnegative positive semi-definite matrix of size $n$. Assume that $Q$ 
has a decomposition $Q=UU^T$, where $U$ is a nonnegative matrix of size 
$n\times r$. A basic fact from linear algebra is that $Q,U$ and $U^T$ have the same rank: $\rank(Q)=\rank(U)=\rank(U^T)=s\leq r$.  

Fix an arbitrary basis $V$ of the column vectors of $Q$, let $B$ be an $n\times s$ matrix corresponding to this basis and let $Q_V$ be the matrix of size $s\times n$ corresponding to the (unique) representation of $Q$ in the basis $V$, that is, $BQ_V=Q$. 
To prove the claim of the theorem, it will suffice to prove the following.

\begin{claim}\label{cl1}
 $Q$ has a non-negative factorization $Q=UU^T$ of inner-dimension at most $r$ if and only if there is an $r'\times s$ matrix $H$, with $r'\le r$, satisfying the following two conditions: 

\begin{description}
	\item (i) $HQ_V$ is a nonnegative matrix,
	\item (ii) $(HQ_V)^T(HQ_V)=Q$.
\end{description}
\end{claim}
\fullonly{
\begin{proof}
$(\Leftarrow)$ Suppose that the conditions (i) and (ii) are satisfied. Let
$U=(HQ_V)^T$.
This matrix is nonnegative and has size of $n\times r$ and thus, $Q$ has cp-rank at most $r'\le r$. 

$(\Rightarrow)$ Now suppose that there is a non-negative factorization $Q=UU^T$, where $U$ is an $n\times r'$-matrix, with $r'\le r$.  
The singular value decomposition of $U$ has the form $U=LSR^T$, where $L,R$ are orthogonal matrices ( $L^{-1}=L^T$ and $R^{-1}=R^T$) of size $n\times n$ and $r'\times r'$, respectively, and $S$ is an $n\times r'$ diagonal matrix. The first $s$ (non-zero) diagonal entries of $S$ are exactly the singular values of $U$ (corresponding to the first $s$ (non-zero) diagonal entries of $Q$). 
We can rewrite $Q$ in the form $Q=(LSR^T)(LSR^T)^T=(LSR^T)(RS^TL^T)$.
Let $N=RS^+L^T$ (an $r'\times n$ matrix), where $S^+$ is the {\em pseudo-inverse} 
of matrix $S$. More precisely, $S^+$ is an $r'\times n$ matrix such that:
\[SS^+={\left( {\begin{array}{*{20}{c}}
   {{I_{s \times s}}} & 0  \\
   0 & 0  \\
 \end{array} } \right)_{n \times n}}
\qquad \text{and}\qquad 
 S^+S={\left( {\begin{array}{*{20}{c}}
   {{I_{s \times s}}} & 0  \\
   0 & 0  \\
 \end{array} } \right)_{r' \times r'}},
\]
where $I_{s\times s}$ is an identity matrix of size $s$.

Define an $r'\times s$ matrix $H=NB$. We now prove that this matrix $T$ satisfies the two conditions (i) and (ii). Indeed, we have:
\begin{align*}
HQ_V=NBQ_V=(RS^+L^T)BQ_V= (RS^+L^T)Q=&\,\, (RS^+L^T)(LSR^T)(RS^TL^T)\\
= & \,\, RS^+SS^TL^T= RS^TL^T = U^T,
\end{align*}
from which both (i) and (ii) follow.
\end{proof}
}
By the claim, to check if $\cpr(Q)\le r$ is equivalent to determining if a system of at most $nr$ linear inequalities and $n^2$ quadratic equations on at most $r^2$ variables is feasible. This decision problem can be solved in $n^{O(r^2)}$ time using quantifier elimination algorithms \cite{BPR96}. Furthermore, a rational approximation to the solution to within an additive accuracy of $\delta$ can be computed in time $\poly(\cL,n^{O(r^2)},\log\frac{1}{\delta})$; see \cite{GV88,R92}. 
\end{proof}

\begin{corollary}\label{cp-rank-approx}
For a real matrix $M$ let $\|M\|_\infty:=\max_{i,j}|M_{ij}|$. Given a rational positive semi-definite matrix $Q$ of size $n$ such that $\cpr(Q)=r$, and $\epsilon>0$, one can compute a rational nonnegative $n\times r$-matrix $\widetilde{U}$ such that $\widetilde{U}\widetilde{U}^T\ge Q$ and $\|Q-\widetilde{U}\widetilde{U}^T\|_\infty\le \epsilon$, in time $\poly(\cL,n^{O(r^2)},\log\frac{1}{\epsilon})$. 
\end{corollary}
\fullonly{
\begin{proof}
Let $s=\rank(Q)$ and $B$ and $Q_V$ be the matrices corresponding to basis $V$ of $Q$, as in the proof of Theorem~\ref{cp-rank}.  Then from this proof, it follows that we can discover in $n^{O(r^2)}$ time that there is an $r\times s$-real matrix $H$ such that conditions (i) and (ii) of claim~\ref{cl1} hold. Furthermore, in time $\poly(\cL,n^{O(r^2)},\log\frac{1}{\delta})$, we can find a rational matrix $\widehat H$, such that $\|H-\widehat H\|_\infty\le \delta$, for any desired accuracy $\delta>0$. Let $B'$ be a non-singular $s\times s$-submatrix of $B$, obtained by selecting $s$ linearly independent rows of $B$, and let $Q'$ be the corresponding $s\times n$ submatrix of $Q$. Then $Q_V=(B')^{-1}Q'$. 
Let us choose 
\begin{equation}\label{delta}
\delta:=\min\left\{\frac{\epsilon}{\|(B')^{-1}\|_\infty\|Q\|_\infty^{3/2}rs^2(4+3s\|Q\|_\infty^{1/2})},\frac{1}{2\cdot\max\{s\|(B')^{-1}\|_\infty,1\}} \right\},
\end{equation}
and set $\delta':=s\delta\| (B')^{-1}\|_\infty$ and $\widetilde H:=\widehat H+\delta'E_{r\times s}B'$, where $E_{r\times s}$ is the $r\times s$-matrix with all-ones. Define $U:=(HQ_V)^T$, $\widehat U:=(\widehat{H}Q_V)^T$, and $\widetilde U:=(\widetilde{H}Q_V)^T$. Then, it follows that 
\begin{eqnarray*}\label{diff1}
\widetilde{U}^T&=&U^T+[\delta'E_{r\times s}+(\widetilde{H}-\widehat H)(B')^{-1}]Q'
\ge U^T+(\delta'-s\delta\| (B')^{-1}\|_\infty) E_{r\times s}Q'=U^T,
\end{eqnarray*}
where the first inequality follows from the fact that $Q'\ge 0$, and the second equality follows from our choice of $\delta'$.  Let us further note that
\begin{eqnarray}\label{diff2}
\widetilde{U}\widetilde{U}^T-UU^T&=&(Q')^T\Delta^TU^T+U\Delta Q'+(Q')^T\Delta^T\Delta Q',
\end{eqnarray}
where $\Delta:=\delta'E_{r\times s}+(\widetilde{H}-\widehat H)(B')^{-1}$. To bound $\|Q-\widetilde{U}\widetilde{U}^T\|_\infty$, we bound the norm of each term in \raf{diff2}. We first note that $\|\Delta\|_\infty\le \delta'+\delta s\|(B')^{-1}\|_\infty=2\delta'$, and thus 
\begin{equation}\label{diff3}
\|U\Delta Q'\|_\infty\leq 2\delta'\|U E_{r\times s}Q'\|_\infty\le 2\delta' r s\|U\|_\infty\|Q\|_\infty\le 2\delta'rs\|Q\|_\infty^{3/2},
\end{equation}
where the last inequality follows from the fact that $\|U\|_\infty\leq \|Q\|_\infty^{1/2}$. Now we write
\begin{eqnarray}\label{diff4}
\Delta^T\Delta&=&(\delta')^2E_{s\times r}E_{r\times s}+\delta' (E_{s\times r}X+X^TE_{r\times s})+X^TX,
\end{eqnarray}
where $X:=(\widetilde H-\widehat H)(B')^{-1}$. From \raf{diff4} and $\|X\|_\infty\le \delta'$ follows the inequality $\|\Delta^T\Delta\|_\infty\le 3r(\delta')^2$, and then the inequality
\begin{equation}\label{diff5}
\|(Q')^T\Delta^T\Delta Q'\|_\infty\leq 3r(\delta')^2\|(Q')^T E_{s\times s}Q'\|_\infty\le 3r(s\delta')^2\|Q\|_\infty^2.
\end{equation}
Using~\raf{diff3} and \raf{diff5} in \raf{diff2}, we get
\begin{equation*}\label{diff6}
\|\widetilde{U}\widetilde{U}^T-Q\|_\infty\le 4\delta'rs\|Q\|_\infty^{3/2}+3r(s\delta')^2\|Q\|_\infty^2\le \epsilon,
\end{equation*}
by our selection~\raf{delta} of $\delta$. Finally note that $\log\frac{1}{\delta}\le \poly(\cL,r\log \frac{1}{\epsilon})$. 
\end{proof}
}

%% file: approx-scheme.tex
\section{Approximation Algorithms for Binary Quadratic Programming}
\label{sec:qp}

Binary programming with quadratic constraints (BQC) is the problem of maximizing (or minimizing) a function of a set of binary variables, subject to a finite set of quadratic constraints. A BQC problem takes the form:

{\centering \hspace*{-18pt}
\begin{minipage}[t]{.47\textwidth}
\begin{alignat}{3}
\textsc{(Packing-BQC)} \quad& \displaystyle \max f(x) \label{PQC}\\
\text{subject to}\quad & \displaystyle x^TQ_ix \le C_i^2,\,\,i\in[m]\nonumber\\
\qquad &x\in\{0,1\}^n,\nonumber
\end{alignat}
\end{minipage}
\,\,\, \rule[-14ex]{1pt}{14ex}
\begin{minipage}[t]{0.47\textwidth}
\begin{alignat}{3}
\textsc{(Covering-BQC)} \quad& \displaystyle \min f(x) \label{CQC}\\
\text{subject to}\quad & \displaystyle x^TQ_ix \ge C_i^2,\,\,i\in[m]\nonumber\\
\qquad &x\in\{0,1\}^n,\nonumber
\end{alignat}
\end{minipage}}

\noindent where $f$ is a nonnegative function, and $Q_i$ is (without loss of generality) an $n\times n$ nonnegative symmetric matrix, for all $i\in[m]:= \{1,2,\ldots,m\}$. 
For simplicity, we fix some useful notation. For any integer $n\ge 1$, we denote by $\mathcal{CP}_n$ the set of completely positive matrices (i.e., with finite cp-rank), and by $\mathcal{CP}^*_n$  the subset of $\mathcal{CP}_n$ that includes matrices of fixed cp-rank. Furthermore, we identify a vector $x\in \{0,1\}^n$ with a subset $S\subseteq [n]$, i.e, write $S=S(x)=\{i\in [n]~|~x_i=1\}$. Hence, for a function $f$ defined on the power set $2^{[n]}$, $f(x)\equiv f(S)$. We shall consider non-negative functions $f$ with a special structure: $f$ is {\it linear} if $f(x):=u^Tx$ for some vector $u\in\RR_+^n$;
$f$ is {\it submodular} if $f(S\cup T)+f(S\cap T)\le f(S)+f(T)$ for all subsets $S,T\subseteq[n]$; $f$ is {\it quadratic}\footnote{Note that a quadratic function $f:2^{[n]}\to\RR_+$ is {\it supermodular}, that is, $f(S\cup T)+f(S\cap T)\ge f(S)+f(T)$ for all subsets $S,T\subseteq[n]$.} if $f(x)=x^TQx+u^Tx$, for some $Q\in\mathcal{CP}_n^*$ and $u\in\RR_+^n$; we refer to the corresponding packing (covering) problems  as \PL, \PS, and \PQ\ (\CL, \CS, and \CQ), respectively. As standard in the literature, we assume that the submodular function $f$ is given by a value oracle that always returns the value $f(S)$ for any given set $S\subseteq [n]$. 

For $\alpha>0$, a vector $x\in\{0,1\}^n$ (or the corresponding set $S\subseteq[n]$) is said to be {\it $\alpha$-approximate solution} for problem \PQC\ (resp., \CQC), if $x$ is a feasible solution satisfying $f(x)\ge \alpha\cdot\OPT$ (resp., $f(x)\le \alpha\cdot\OPT$), where $\OPT$ is the value of an optimal solution; for simplicity, if $\alpha=1-\epsilon$ (resp., $1+\epsilon$), for $\epsilon>0$, $x$ will be called {\it $\epsilon$-optimal}.

It can be seen that the problem (BQC) is NP-hard as it includes among others the {\it multi-dimensional Knapsack packing and covering} problems as special cases, when we set $f$ to be a linear function and set each matrix $Q_i$ to be a rank-one non-negative matrix. A {\it fully polynomial-time approximation scheme} (FPTAS)\footnote{A PTAS is an algorithm that runs in time polynomial in the input size $n$, for every fixed $\epsilon$, and outputs an $\epsilon$-optimal solution; an FPTAS is a PTAS where the running time is polynomial in $\frac{1}{\epsilon}$; a QPTAS is similar to a PTAS but the running time is quasi-polynomial (i.e., of the form $n^{\polylog n}$), for every fixed $\epsilon$.} for the knapsack problem was given in \cite{iba-kim:j:fast-approximation-algorithm-for-knapsack-subset-problem}, and a PTAS for the $m$-dimensional knapsack problem was given in \cite{cha-hir-won:j:approximation-algorithm-for-some-generalized-knapsack-problem,fri-cla:j:approximation-algorithm-for-m-dimensional-knapsack}, and these results are best possible, assuming P$\neq$NP. When $f$ is submodular, it was shown that there exist constant factor approximation algorithms for the packing problem, subject to a constant number of knapsack constraints (e.g., \cite{svi:c:note-on-maximizing-submodular-function-knapsack-constraint,kul-sha-tam:c:maximizing-submodular-set-function-subject-to-multiple-linear-constraint}).
We extend these results as follows. 
%
%

\begin{theorem}\label{thm:lin}
{\PL} admits a PTAS when $m$ and each $\cpr(Q_i)$ is fixed. 
\end{theorem}

\begin{theorem}\label{thm:sub}
For any $\epsilon>0$, there is a $(1-\epsilon)\alpha$-approximation algorithm for {\PS} when $m$ and each $\cpr(Q_i)$ is fixed, where 
$\alpha=(1-\frac{1}{e}-\epsilon)$, if the objective $f$ is monotone, and $\alpha=\frac{1}{4} - \epsilon$, otherwise. 
\end{theorem}

\begin{theorem}\label{thm:cov}
{\CL}  admits a QPTAS when $m=1$, and $\cpr(Q)$ is fixed. 
\end{theorem}
 
The next sections are devoted to present two methods for designing approximation algorithms for \textsc{Packing-BQC}.
The first method relies on the polynomial solvability of the convex relaxation, whereas the second one takes into account the geometric properties of the region of feasible solutions. The first method works for the case of linear objective function and allows also for an additive linear term in each quadratic constraint; it can also be extended to handle more general classes of objective functions such as quadratic functions with fixed cp-rank, multiplicative functions, and sums of ratios of linear functions, as well as to packing constraints involving the $\ell_p$-norm for $p\ge 1$. However, we are not able to show that this method works for submodular objective functions. The second method works for the latter case, but it does not seem to have the flexibility of the first method in handling more general objective functions and/or additional linear terms in the quadratic constraints. 

\sodaonly{
For the minimization problem \textsc{\CL}, we use the idea of geometric partitioning combined together with a greedy covering approach to arrive at the result of Theorem~\ref{thm:cov}; see Section \ref{sec:min}. We prove Theorem~\ref{thm:sub} in Section~\ref{sec:sub} and Theorem~\ref{thm:lin} for the linear objective case in Section~\ref{lin-convex}.} 
\fullonly{
For the minimization problem \textsc{\CL}, we use the idea of geometric partitioning combined together with a greedy covering approach to arrive at the result of Theorem~\ref{thm:cov}; see Section \ref{sec:min}. We prove Theorem~\ref{thm:sub} in Section~\ref{sec:sub} and Theorem~\ref{thm:lin} for the linear objective case in Section~\ref{lin-convex}; we extend the result to the quadratic objective case in Section~\ref{sec:appl}.} 

\fullonly{As a technical remark, we note that the decomposition of each $Q_i$ as given by Corollary~\ref{cp-rank-approx} involves an error term $\Delta_i$. To simplify the presentation, we will assume first that an exact decomposition $Q_i=U_iU_i^T$ is given, then show in Section~\ref{err} to see how to deal with the error term.  }

\sodaonly{As a technical remark, we note that the decomposition of each $Q_i$ as given by Corollary~\ref{cp-rank-approx} involves an error term $\Delta_i$. To simplify the presentation, we will assume that an exact decomposition $Q_i=U_iU_i^T$ is given (see Appendix \ref{apx:error} to see how to deal with the error term).}

\subsection{Linear Objective Function and Convex Programming-based Method}\label{lin-convex}
In this section, we present a polynomial-time approximation scheme for \textsc{Packing-BQC} with linear objective function. This problem has the form \raf{PQC} with $f(x)=u^Tx$,
where $Q_i\in\mathcal{CP}^*_n$ and $u\in\RR_+^n$. For convenience, we will sometimes write $u(S):=\sum_{k\in S} u_k$, for a set $S\subseteq[n]$.
The method we use here is based on solving the convex relaxation obtained by replacing the constraint $x\in \{0,1\}^n$ by the weaker one $x\in[0,1]^n$. The crucial point of our method is that an upper bound on the optimal value can be computed easily by solving (approximately) the relaxed problem. The obtained solution to the convex program defines a point in a certain polytope, which is then rounded to an integral solution without losing much on the quality of the solution. 
The details of our method is described in Algorithm \ref{Lin-QC-PTAS}.

\fullonly{
\begin{algorithm}[!htb]
	\caption{ \PLPTAS$(u,\{Q_i,C_i\}_{i\in[m]},\epsilon)$} \label{Lin-QC-PTAS}
\begin{algorithmic}[1]
\Require Utility vector $u\in\RR_+^n$; matrices $Q_i\in\mathcal{CP}_n^*$, capacities $C_i\in\RR_+$, for $i\in[m]$; and accuracy parameter $\epsilon$
\Ensure An $\epsilon$-optimal solution $S$ to \PL
\State $v\leftarrow (0,\ldots,0)$
\State For all $i\in [m]$, decompose $Q_i$ as $Q_i=U_iU_i^T$, where $U_i\in\QQ^{n\times r_i}$ \Comment{{\em $r_i=\text{cp-rank}(Q_i)$}}\label{s0}
\State $\bar r\leftarrow\sum_{i=1}^m(r_i+1);\, \lambda\leftarrow\frac{\bar r}{\epsilon}$
\For{each subset $\mathcal{U}$ of $[n]$ of cardinality at most $\lambda$}
\State $\mathcal{V}\leftarrow\{k\in [n]\setminus \mathcal{U}~|~u_{k}>\min\{u_{k'}|k'\in \mathcal{U}\}\}$
\If{$\bone_\cU^TQ_i[\cU;\cU]\bone_\cU\le  C_i^2$ for all $i\in[m]$}\label{condition} \Comment{(CP1$[\cU]$) is feasible}
\State Let $x^*$ be an $\epsilon$-optimal solution to the convex program (CP1$[\cU]$)\label{s11}
\State $t_i\leftarrow U_i^T[*;N]x^*_N$; $t_i'\leftarrow \bone_\cU^TQ_i[\cU;N]x^*_N$ for all $i\in [m]$\Comment{{\em $t_i$ is a vector of dimension $r_i$}}\label{s2}
\State Find a BFS $y$ of the polytope $\cP(\cU)$ such that $u^Ty\ge u^Tx^*_N$:
\State $\overline x\leftarrow \{(\overline x_1,\ldots, \overline x_n)|\overline x_k=\left\lfloor y_k\right\rfloor\text { for }k\in [n]\setminus(\cU\cup\cV),~ x_k=1, \text{ for }k\in\cU,~ x_k=0, \text{ for }k\in\cV\}$ ~~~~~~~~~
\Comment{{\em rounding down $y_k$}}\label{s3}
\If{$u^T\overline x>u^Tv$}
\State $v\leftarrow \overline x$
\EndIf
\EndIf
\EndFor
\State \Return $S(v)$
\end{algorithmic}
\end{algorithm}}

Let $\epsilon>0$ be a fixed constant, and define $\lambda=\frac{\bar r}{\epsilon}$, where $\bar r=\sum_{i=1}^m(r_i+1)$. The idea, extending the one in \cite{fri-cla:j:approximation-algorithm-for-m-dimensional-knapsack}, is to try to guess $\lambda$ items of highest utility in the optimal solution. This can be done by considering all possibilities for choosing a set of cardinality at most $\lambda$. We denote by $\mathcal{X}\subseteq 2^{[n]}$ the set of all such sets. Note that the size of $\mathcal{X}$ is bounded by $O(n^{\lambda})$ and thus is polynomial in the size of input for every constant $\lambda$. 

For each $\cU\in \cX$, we define a set $\cV$, which contains all items that are {\it not} in the optimal solution, given that $\cU$ is a subset of the optimal solution (these are the items $k\in[n]\setminus\cU$, with $u_k>u_{k'}$ for some $u_{k'}\in\cU$). Let $N:=[n]\setminus(\cU\cup\cV)$.
Using polynomial-time algorithms for convex programming (see, e.g., \cite{NT08}), we can find an $\epsilon$-optimal solution\fullonly{\footnote{In fact, such algorithms can find a feasible solution $x^*$ to the convex relaxation such that $u^Tx^*\ge \opt-\delta$, in time polynomial in the input size (including the bit complexity) and $\log\frac{1}{\delta}$, where $\opt$ is the value of the fractional optimal solution. We may assume that $\bone_k^TQ_i\bone_k\le C_i^2$, for all $i\in[m]$, where $\bone_k$ is the $k$th unit vector; otherwise item $k$ can be removed form the (\PL) problem. This implies that $\opt\ge\OPT\ge\bar u:=\max_k{u_k}$. Now setting $\delta$ to $\epsilon \cdot\bar u$ assures that $u^Tx^*\geq(1-\epsilon)\opt$.}} to the convex relaxation (CP1$[\cU]$) defined as follows:
\begin{align*}
\textsc{(CP1$[\cU]$)}    & \qquad \displaystyle \max u^Tx\\
\text{subject to} & \qquad\displaystyle x^TQ_ix \le C_i^2,\,\,i\in[m],\\
&\qquad x_k=1  \text{ for } k\in \mathcal{U}; x_k=0 \text{ for } k\in \mathcal{V},\\
&\qquad x_k\in[0,1] \text{ for } k\in N.
\end{align*}

We then define a polytope $\cP(\cU)\subseteq[0,1]^{N}$ by replacing each constraint $x^TQ_ix \le C_i^2$ by the two constraints $U_i^T[*;N]y\le t_i$ and $\bone_\cU^TQ_i[\cU;N]y\le t_i'$, for all $i\in[m]$, where $t_i:= U_i^T[*;N]x^*_N$; $t_i':= \bone_\cU^TQ_i[\cU;N]x^*_N$; here, $x_\cU$ denotes the restriction of the vector $x\in[0,1]^n$ to the set of components indexed by $\cU$, and $Q_i[\cU;N]$ denotes the restriction of the matrix $Q_i$ to the columns and rows indexed by the sets $\cU$ and $N$, respectively; similarly, $U_i^T[*;N]$ means the restriction of $U_i^T$ to the set of columns defined by $N$:
$$
\cP(\cU):=\{y\in[0,1]^{N}|~U_i^T[*;N]y \le t_i,~\bone_\cU^TQ_i[\cU;N]y\le t_i',\text{ for }i\in[m]\}.
$$ 
 We can find a basic feasible solution (BFS) $y$ in this polytope  such that $u^Ty\ge u^Tx^*_N$, by solving at most $n$ linear systems (see standard texts on Linear Programming, e.g., \cite{S86}). 
Next, an integral solution $\overline x$ is obtained from $y$ by dropping all fractional components to $0$ and setting $x_k\in\{0,1\}$ according to the assumption $k\in\cU\cup\cV$. For each set $\U\in \cX$ 
such that (CP1$[\cU]$) is feasible, the algorithm outputs an integral solution $\overline x$. Let $v$ be the solution of maximum value amongst all such solutions $\overline x$. The algorithm produces a set $S(v)$  of items that corresponds to $v$.
\sodaonly{
\begin{lemma}
\label{theo:ptas-lin-obj}
Assume $Q_i\in\mathcal{CP}_n^*$ for all $i\in[m]$ and $m=O(1)$. Then, for any fixed $\epsilon>0$, the algorithm described above runs in polynomial time and produces an $2\epsilon$-optimal solution to the input instance.
\end{lemma}
}
\fullonly{
\begin{lemma}
\label{theo:ptas-lin-obj}
Assume $Q_i\in\mathcal{CP}_n^*$ for all $i\in[m]$ and $m=O(1)$. Then, for any fixed $\epsilon>0$, Algorithm \ref{Lin-QC-PTAS} runs in polynomial time and produces an $2\epsilon$-optimal solution to the input instance.
\end{lemma}

\begin{proof}
From the argument above, it can be easily seen  that the running time of the Algorithm \ref{Lin-QC-PTAS} is polynomial in size of the input, for any fixed constant $\epsilon$. We now argue that the solution $S$ returned by the algorithm is $2\epsilon$-optimal. Indeed, let $S^*$ be the optimal solution to {\PL} of utility $u(S^*)=\OPT$. If $\lambda\ge |S^*|$, then $S$ is an exact optimum solution to the {\PL} instance. Suppose that $\lambda< |S^*|$. Let $\cU$ be the set of highest-utility $\lambda$ items in $S^*$. Then, the feasibility of $S^*$ for (CP1$[\cU]$) guarantees that $\bone_\cU^TQ_i[\cU;\cU]\bone_\cU\le  C_i^2$, and hence an $\epsilon$-optimal solution $x^*$ for (CP1$[\cU]$) is found in step~\ref{s11} of the algorithm. Assume w.l.o.g. that $\bone_\cU^TQ_i[\cU;\cU]\bone_\cU\leq  C_i^2$. From the feasibility of $x^*$ for (CP1$[\cU]$) and the factorization of $Q_i$, it follows that, for all $i\in[m]$,  
\begin{equation*}\label{con1}
(x^*_N)^TU_i[N;*]U_i^T[*;N]x_N^*+2\cdot\bone^T_\cU Q_i[\cU;N]x^*_N\le C_i^2-\bone_\cU^TQ_i[\cU;\cU]\bone_\cU,
\end{equation*}
implying that
\begin{equation}\label{con2}
t_i^2+2t_i'\le C_i^2-\bone_\cU^TQ_i[\cU;\cU]\bone_\cU.
\end{equation}
The definition of $t_i,t_i'$ implies that $x_N^*\in\cP(\cU)$. Thus, there is a BFS $y$ of $\cP(\cU)$ such that $u^Ty\ge u^Tx^*_N$. Let $\overline x\in\{0,1\}^n$ be the rounded solution obtained from $y$ in step~\ref{s3}. Since any BFS of $\cP(\U)$ has at most $\bar r$ fractional components and $u_k\le \min_{k'\in\cU}u_{k'}$ for all $k\in N$, it follows that 
\begin{eqnarray*}
u^T\overline x&=&u^T\overline x_N+u^T\bone_{\cU}\ge u^Ty-\frac{\bar r\cdot u^T\bone_{\cU}}{|\cU|}+u^T\bone_{\cU}\\
&\ge&u^Tx_N^*+(1-\dfrac{\bar r}{\lambda})u^T\bone_{\cU}\ge (1-\epsilon)u^Tx^*\ge(1-\epsilon)^2\OPT.
\end{eqnarray*}
It remains to argue that $\overline x$ is feasible for \PL. 
Since $\overline x_N\in\cP(\cU)$, we have
\begin{eqnarray*}
\overline x^TQ_i\overline x&=&\overline x_N^TU_i[N;*]U_i^T[*;N]\overline x_N+2\cdot\bone^T_\cU Q_i[\cU;N]\overline x_N+\bone_\cU^TQ_i[\cU;\cU]\bone_\cU\\
&\le& t_i^2+2t_i'+\bone_\cU^TQ_i[\cU;\cU]\bone_\cU\le C_i^2,
\end{eqnarray*}  
where the last inequality follows from \raf{con2}.
This completes the proof. 
\end{proof}
}
%
%
\fullonly{\begin{remark}
\label{rem:extension}
Note that the result presented in this section can be extended easily to the following problem:
\begin{align*}
\textsc{} \qquad& \displaystyle \max u^Tx \label{}\\
\text{subject to}\qquad & \displaystyle x^TQ_ix +q_i^Tx\le C_i^2,\,\,i\in[m]\\
\qquad & \displaystyle Ax\le b,\\
\qquad & x\in\{0,1\}^n,
\end{align*}
where $Q_i\in \mathcal{CP}^*_n$, $q_i\in \RR_+^n$, for all $i\in[m]$; $A\in\RR_+^{d\times n}$, $b\in\RR_+^d$; $m,d$ are constants. 
\end{remark}
}

\subsection{Submodular Objective Function and Geometry-based Method}\label{sec:sub}
In this section, we present a constant factor approximation algorithm for \PS. This problem has the form \raf{PQC} where $f(x)$ is a nonnegative submodular function, $Q_i\in \mathcal{CP}^*_n$, and $m$ is constant. 

If one attempts to use the technique of the previous section, (s)he runs into the difficulty of how to relax the objective function; the two well-known relaxations (see, e.g., \cite{D09} for a recent survey) are: the {\it Lov\'{a}sz extension} which is known to be {\it convex} \cite{GLS88}, and the {\it multi-linear extension} \cite{V08} which is in general {\it neither} convex {\it nor} concave.
While it is not known how to solve the maximization problems for the relaxations corresponding to these two types of functions over a polytope in polynomial time, it is known \cite{V08,VCZ11} how to approximate within a constant factor the maximum of the multi-linear extension over a packing polytope. It is not clear if this result can be extended to the case when the feasible region is described by convex quadratic constraints of the form $x^TQ_ix\le C_i^2$, where $Q_i\in\mathcal{CP}_n^*$. While this remains an interesting open question, we will rely here, instead, on a geometric-based approach that uses the results in \cite{VCZ11,kul-sha-tam:c:maximizing-submodular-set-function-subject-to-multiple-linear-constraint} as a black-box.  
     The key idea is to make use of the geometry of the problem to reduce it into a multi-dimensional knapsack problem, which  can be solved using enumeration and dynamic programming for the linear objective case, or LP-rounding for the submodular case. 

Let $I=(f,\{Q_i,C_i\}_{i\in[m]})$ be an instance of the problem \PS, and $\epsilon>0$ be a fixed constant. We will construct a polynomial-time algorithm which produces a constant factor-approximate solution to the instance $I$. Write $Q_i=U_iU_i^T$, for some $U_i\in\QQ^{n\times r_i}_+$, where $r_i=\cpr(Q_i)$. For $k\in[n]$, define the vector $q^i_k\in\QQ_+^{r_i}$ to be the $k$th column of $U_i^T$. 
Since $\bar r\triangleq\max_{i\in[m]}{r_i}$ is bounded by a constant, we may assume without loss of generality that  $\epsilon<\frac{1}{4\bar r}$.
For $r\in\ZZ_+$ and $C\in\RR_+$, define $\cB(r,C)\triangleq\{\nu\in\RR^r_+:~\|\nu\|_2\le C\}$ to be the non-negative sector of the ball in $\RR^r$ of radius $C$ centered at the origin. 
Then problem \PS\ amounts to finding an $S\subseteq[n]$, maximizing $f(S)$, such that $\sum_{k\in S}q^i_k\in\cB(r_i,C_i)$ for all $i\in[m]$.

Given a feasible set $T \subseteq [n]$ (that is, $\|d_T\|_2\le C_i$ for all $i\in[m]$), we define $\cR_T^i$ as the conic region bounded as the following:
\begin{equation}\label{region}
\cR_T^i \triangleq \Big\{\nu \in\RR^{r_i}_+ :~  \|\nu\|_2 \le C_i^2,~ \nu \ge q_T^{i}\Big\},
\end{equation} 
where $q_T^i\triangleq\sum_{k \in T} q^i_k$. Write $q^i_k\triangleq(q^{ij}_k:~j\in[r_i])\in\RR^{r_i}$. Given $\cR_T^i$, we define an $r_i$-dimensional grid in the region $\cR_T^i$ by interlacing equidistant $(r_i-1)$-dimensional parallel hyperplanes with inter-separation $\frac{\epsilon}{2r_i}w^{ij}_T$, for each $j\in[r_i]$, and with the $j$th principal axis as their normal, where 
\begin{equation}\label{sqrt}
w^{ij}_T \triangleq \sqrt{C_i^2 -\sum_{j'\neq j} (q^{ij'}_T)^2} - q^{ij}_T
\end{equation} 
is the distance along the $j$th axis from $q_T^i$ to the boundary of the ball\fullonly{\footnote{To maintain {\it finite precision}, we need to {\it round down} $w_T^{ij}$ in \raf{sqrt} to the nearest integer multiple of $\frac{1}{D}$, where $D$ is the common denominator of all the rational number $q_k^{ij}$; this should be done also for any computations involving square roots in the sequel. This introduces an error that can be dealt with in a similar way as we deal with the error caused by the $\Delta_i$'s. To keep the presentation simple, we will assume infinite precision.}}. 
	For $j\in[r_i]$, let $H^{ij}$ be such hyperplane that is furthest from the origin, perpendicular to the $j$th principal direction, and let $H^{ij}_+$ be the half-space defined by $H^{ij}$ that includes the origin. This grid defines at most $r_i(\frac{2r_i}{\epsilon})^{r_i-1}$ lines that intersect the spherical boundary of region $\cR_T^i$ at a set of points $P_T^i(\epsilon)$. Let us call a cell (that, is a hypercube) of the grid {\it boundary} if it overlaps with the surface of the ball $\cB(r_i,C_i)$ and denote by $\cB_T^i$ the set of all such boundary cells. The convex hull of the set of points $P_T^i(\epsilon)$ defines a polytope $\cQ_T^i$ with the following properties; the first two are easily verifiable; the last one follows form the {\it upper bound Theorem} \cite{M70} and the algorithms for convex hull computation in fixed dimension (e.g. \cite{B93}):
\begin{itemize}
\item[(I)] All points in $P_T^i(\epsilon)$ are vertices of $\cQ_T^i$.
\item[(II)] Let $\cH_T^i$ be the set of half-spaces including the origin and defined by the  facets of $\cQ_T^i$. Then for each $H_+\in\cH_T^i$, the facet defining $H_+$ is the convex hull of a set of vertices from $P_T^i(\epsilon)$ that are contained completely inside a boundary cell, and its normal is a {\it nonnegative} vector.   
\item[(III)] The number of facets of $\cQ_T^i$ is at most $(\frac{2r_i}{\epsilon})^{r_i^2/2}$ and they can be found in time $O((\frac{2r_i}{\epsilon})^{r_i-1}(r_i^2\log\frac{r_i}{\epsilon}+(\frac{2r_i}{\epsilon})^{((r_i-1)^2+1)/2}))$.  
\end{itemize}
Let $\cP_T^i(\epsilon)$ be the polytope defined by the intersection
$$
\left(\bigcap_{H_+\in\cH_T^i}H_+\right)\cap\left(\bigcap_{j\in[r_i]}H_+^{ij}\right)\cap\{\nu\in\RR^{r_i}:~\nu\ge 0\},
$$ and denote by $m_T^i(\epsilon)$ the number of facets of $\cP_T^i(\epsilon)$. By construction, $m_T^i(\epsilon)\le(\frac{2r_i}{\epsilon})^{r_i^2/2}+2r_i$. 

Consider a collection of feasible subsets $\cT\triangleq\{T^1,\ldots,T^m\}$, $T^i\subseteq[n]$, to \PS. We define the following approximate version of problem \PS, which is obtained by approximating the ball sectors $\cB(r_i,C_i)$ by inscribed polytopes.
\begin{eqnarray*}
\textsc{(PQP$_\cT$)} \qquad& \displaystyle \max_{x \in \{0, 1 \}^n} f(x) \label{PQP}\\
\text{subject to}\qquad & \displaystyle \sum_{k\in [n]}q^i_k x_k \in \cP_{T^i}^i(\epsilon),~~~\text{ for all } i\in[m]. \label{CPQP}
\end{eqnarray*}

Given two vectors $\mu,\nu\in\RR^{r_i}$, we denote the (vector) projection of $\mu$ on $\nu$ by ${\tt Pj}_\nu(\mu) \triangleq \frac{\nu}{\|\nu\|_2^2}\mu^{T}\nu$. Given the polytope $\cP_T^i(\epsilon)$, we define a set of $m_T^i(\epsilon)$ vectors $\{ \sigma^{T,i}_\ell\}$ in $\RR^{r_i}$, each of which is perpendicular to a facet of $\cP^i_T(\epsilon)$, starting at the origin, and ending at the facet.

For a collection of feasible subsets $\cT\triangleq\{T^1,\ldots,T^m\}$ to \PS, we define a submodular function maximization problem subject to $\bar m$ knapsack constraints
based on $\{ \sigma^{T^i,i}_\ell\}$, where $\bar m\triangleq\sum_{i\in[m]}m_{T^i}^i(\epsilon)$:
\begin{eqnarray*}
\textsc{($\bar m$DKS-Sub$\{ \sigma^{T^i,i}_\ell\}$)} \qquad & \displaystyle \max_{x \in \{0, 1 \}^n}  f(x) \label{(mDKS)}\\
\text{ subject to }\qquad & \displaystyle \sum_{k\in[n]} \|\Pij(q^{ij}_k) \|_2 x_k \le  \|\sigma^{T^i,i}_\ell\|_2, \quad \text{for all }\ell= 1, \ldots, m_{T^i}^i(\epsilon). \label{qp-Cm-1}
\end{eqnarray*}
The following lemma follows straightforwardly from the convexity of the polytopes $\cP_{T^i}^i(\epsilon)$.
 
\begin{lemma}\label{qp-lem-proj}
Given a collection of feasible solutions $\cT\triangleq\{T^1,\ldots,T^m\}$ to \PS,  problems {\sc PQP}$_\cT$ and {\sc $\bar m$DKS-Sub}$\{ \sigma^{T^i,i}_\ell\}$ are equivalent.
\end{lemma}

Our 
approximation algorithm for \PS\ is described in Algorithm \ref{PSAPPROX} below, which enumerates every collection of sets $\cT=\{T^1,\ldots,T^m\}$ s.t. $|T^i|\le\frac{1}{\epsilon}$, then finds a near optimal solution for \textsc{PQP$_\cT$} using an 
approximation algorithm for {\sc $\bar m$DKS-Sub} (e.g., the algorithm in \cite{kul-sha-tam:c:maximizing-submodular-set-function-subject-to-multiple-linear-constraint}). 

\begin{algorithm}[!htb]
	\caption{ \PSAPPROX$(f,\{Q_i,C_i\}_{i\in[m]},\epsilon,\alpha(\epsilon))$} \label{PSAPPROX}
\begin{algorithmic}[1]
\Require A submodular function $f:\{0,1\}^n\to\RR_+$; matrices $Q_i\in\mathcal{CP}_n^*$, capacities $C_i\in\RR_+$, for $i\in[m]$; accuracy parameter $\epsilon$; and an $\alpha(\epsilon)$-approximation algorithm for {\sc $\bar m$DKS-Sub} with $\bar m=O(1)$
\Ensure $(1-\epsilon)^m\alpha(\epsilon)$-approximate solution $\hat{S}$ to \PS
\State For all $i\in [m]$, decompose $Q_i$ as $Q_i=U_iU_i^T$, where $U_i\in\QQ_+^{n\times r_i}$ \Comment{{\em $r_i=\text{cp-rank}(Q_i)$}}\label{s0}
\State $\hat{S} \leftarrow \varnothing$
\For{each  collection of sets $\cT=\{T^1,\ldots,T^m\}$ s.t. $T^i\subseteq[n]$ and $|T^i|\le\frac{1}{\epsilon}$}\label{qp-ss1}
  \State Set $q_{T^i}^i\leftarrow\sum_{k \in T^i} q^i_k$, and define the corresponding vectors $\{ \sigma^{T^i,i}_\ell\}$
  \State Obtain an $\alpha(\epsilon)$-approximate solution $S$ to {\sc $\bar m$DKS-Sub}$\{ \sigma^{T^i,i}_\ell\}$ \label{s1} 
  \If{$u(\hat{S}) < u(S)$}
  \State $\hat{S} \leftarrow S$
  \EndIf
\EndFor
\State \Return $\hat{S}$
\end{algorithmic}
\end{algorithm}

\begin{theorem}\label{QP-t2}
For any fixed $\epsilon>0$, if the problem of maximizing a submodular function subject to multiple knapsack constraints can be approximated within some constant factor $\alpha(\epsilon)$, then Algorithm \ref{PSAPPROX} runs in polynomial in size of the input and finds a $(1-\epsilon)^m\alpha(\epsilon)$-approximate solution to \PS. 
\end{theorem}
\begin{proof}
It is easy to see that Algorithm~\ref{PSAPPROX} runs in polynomial time in the number of items $n$. Let $S^\ast$ be an optimal solution. To establish the approximation ratio for Algorithm~\ref{PSAPPROX}, we explicitly construct a set $S \subseteq S^\ast$ in Lemma~\ref{qp-main-lem} below such that $u(S)\ge(1-\epsilon)^mu(S^\ast)$, and  $S$ is feasible to {\sc PQP$_\cT$} for some collection $\cT=\{T^1,\ldots,T^m\}$ s.t. $T^i\subseteq[n]$ and $|T^i|\le\frac{1}{\epsilon}$. By Lemma \ref{qp-lem-proj}, invoking the $\alpha(\epsilon)$-approximation algorithm for {\sc $\bar m$DKS-Sub$\{\sigma_\ell^{T^i,i}\}$} gives  an $\alpha(\epsilon)$-approximation $\widehat S$ to {\sc PQP$_\cT$}. Then it follows that 
\begin{equation}
u( \widehat S ) \ge \alpha(\epsilon) u(S) \ge (1-\epsilon)^{m}\alpha(\epsilon) \OPT.
\end{equation}
The proof is thus completed by Lemmas~\ref{qp-main-lem}-\ref{sub-pro} below.
\end{proof}

\sodaonly{
\begin{lemma}\label{qp-main-lem}
Consider a feasible solution $S^\ast$ to \PS. One can construct a subset $S\subseteq S^\ast$ and a collection $\cT=\{T^1,\ldots,T^m\}$, such that  ~~~
(i) $T^i\subseteq S^\ast$ and $|T^i|\le\frac{1}{\epsilon}$; ~~
(ii) $S$ is a feasible solution to {\sc PQP}$_\cT$ and $u(S)\ge(1-\epsilon)^m u(S^\ast)$.  
\end{lemma}
}
\fullonly{
\begin{lemma}\label{qp-main-lem}
Consider a feasible solution $S^\ast$ to \PS. Algorithm~\ref{QC-Construct} returns a subset $S\subseteq S^\ast$ and a collection $\cT=\{T^1,\ldots,T^m\}$, such that  ~~~
(i) $T^i\subseteq S^\ast$ and $|T^i|\le\frac{1}{\epsilon}$; ~~
(ii) $S$ is a feasible solution to {\sc PQP}$_\cT$ and $u(S)\ge(1-\epsilon)^m u(S^\ast)$.  
\end{lemma}
}
\fullonly{
\begin{proof}

Starting from $S^\ast$, where $\sum_{k\in S^\ast}q_k^1\in\cB(r_1,C_1)$, Algorithm~\ref{QC-Construct} first finds sets $T^1,S^1\subseteq S^*$ such that $|T^1|\le \frac{1}{\epsilon}$, $u(S^1)\ge (1-\epsilon)u(S^*)$, and the set of vectors in $S^1$ can be packed inside the polytope $\cP^1_{T^1}(\epsilon)$. 
 
The way for finding such sets is as follows. Let $\overline\ell$ and $T_{\overline\ell}$ be the values of $\ell$ and $T$ at the end of the repeat-until loop (line~\ref{qp-ss1-}). The algorithm first constructs a nested sequence $T_0\subset T_1 \subset \cdots \subset T_{\overline\ell}$, such that
a vector $q_k^1$ is included in each iteration if it has a "large" component $q^{1j}_k$ for some $j\in[r_1]$. The iteration proceeds until a sufficiently large number of vectors have been accumulated (namely, $|T_{\overline\ell}|\ge \frac{1}{\epsilon}$), or no vectors with large components remain. At the end of the iteration, if the condition in line~\ref{qp-sss1} holds, then set $T^1=T_{\overline\ell}$; otherwise, the algorithm finds a subset of $S^\ast$  that belongs to $\cP^i_{T_{\overline\ell}}(\epsilon)$. To do so, the set $S^*\backslash {T_{\overline\ell}}$ of remaining vectors is partitioned into at least $\frac{1}{\epsilon}-1$ groups such that, along one of the principal axes, the total sum in each group is "large". Then some vector or group of vectors is dropped, ensuring that the set $S^1$ of remaining vectors can be packed inside $\cP^1_{T^1}(\epsilon)$.

We now have to consider two cases (line~\ref{qp-sss2}): (i) $|T_{\overline\ell}|\ge\frac{1}{\epsilon}$, or (ii) $S_{\overline\ell}=\varnothing$. For case (i), the algorithm proceeds to line\ref{qp-sss2} - combining the vectors in $S\backslash {T_{\overline\ell}}$ into one group $V_1$. For case (ii), the set $S\backslash {T_{\overline\ell}}$ can be partitioned into at least $\frac{1}{\epsilon}-1$ groups $V_1,\ldots,V_h$ due to Lemma \ref{qp-lem-pack}, where each group $V_p$, $p\in[h]$, has a large total component along one of principal axes (precisely, greater than $\frac{\epsilon}{2r_1} w^{1j}_{T_{\overline\ell}}$ for some $j\in [r_1]$). We define $S^1$ by deleting some vector or group of vectors from $S$ (lines \ref{qp-sss4} and \ref{qp-sss5}). Hence, in both cases, $S^*=\bigcup_{k\in T_{\overline\ell}}\{k\}\cup \bigcup_{p\in [h]}V_p$ induces a partition of $S^*$ into $|T_{\overline\ell}|+h\ge \frac{1}{\epsilon}$ subsets. By Lemma \ref{sub-pro} below, we have
\[
f(S^1) \ge (1 - \frac{1}{|T_{\overline\ell}|+h})f(S^*)\ge (1 -\epsilon)f(S^*).
\]

Note that removing any one vector $k\in T_{\overline\ell}$ or group $V_p$ from $S^*$ will insure that the resulting set $S^1$ satisfies $q_{S^1}^1\in\cP^1_{T_{\overline\ell}}(\epsilon)$, since the lengths $w^{1j}_{T_{\ell}}$ are monotone decreasing in $\ell$, for all $j\in[r_1]$, and the dropped vector or group of vectors has total length at least $\frac{\epsilon}{2r_i} w^{1j}_T$ along one of the principal directions $j\in[r_1]$; on the other hand, the boundary cell that contains the original point $q_{S^*}^1$ has length at most $\frac{\epsilon}{2r_i} w^{1j}_T$ along the $j$th direction. It follows from property (II) stated above that reducing $j$th component of $q_{S^*}^1$ by $\frac{\epsilon}{2r_i} w^{1j}_T$ guarantees that the resulting vector lies in the polytope $\cP^1_{T_{\overline\ell}}(\epsilon)$.

Now starting from $S^1$ which satisfies $\sum_{k\in S^1}q_k^2\in\cB(r_2,C_2)$, we find sets     
$T^2,S^2\subseteq S^1$ such that $|T^2|\le \frac{1}{\epsilon}$, $f(S^2)\ge (1-\epsilon)f(S^1)\ge (1-\epsilon)^2f(S^\ast)$ and the set of demands in $S^2$ can be packed inside the polytope $\cP^2_{T^2}(\epsilon)$, and so on. Finally, we find the set $\{T^1,\ldots,T^m,S'\}$ that satisfies the lemma.
The details are given in Algorithm~\ref{QC-Construct}.\end{proof}
\begin{algorithm}[!htb]
\caption{{\sc QC-Construct}$(f,\{q^i_k\}_{k\in S^\ast},~\{C_i\}_{i\in[m]},\epsilon)$}
\label{QC-Construct}
\begin{algorithmic}[1]
\Require  A submodular function $f:\{0,1\}^n\to\RR_+$; vectors $\{q^i_k\}_{k\in S^\ast,~i\in[m]}$; capacities $\{C_i\}_{i\in[m]}$; accuracy parameter $\epsilon$
\Ensure A pair $(\cT=\{T^1,\ldots,T^m\},S)$ of a collection of sets $T^i\subseteq S^\ast$ and a subset $S\subseteq S^\ast$
\State $S\leftarrow S^\ast$
\For{$i=1,\ldots,m$}  
\State $\ell\leftarrow0$; $T\leftarrow\varnothing$
\Repeat \Comment{{\em Find a subset of large vectors $T$ w.r.t. to the $i$th constraint}}
  \State $\ell\leftarrow\ell+1$
    \State $S_\ell\leftarrow\{k\in S\setminus T\mid \exists j\in[r_i]: q_k^{ij}>\frac{\epsilon}{2r_i} w^{ij}_{T}\}$, where $w^{ij}_{T}$ is given by \raf{sqrt}
    \State $T\leftarrow T\cup S_\ell$
  \Until{ $|T|\ge \frac{1}{\epsilon}$ \  or \   $S_\ell=\varnothing$ \   or  \  $ S\backslash {T}=\varnothing$}\label{qp-ss1-}   

\If{$q_S^i\in \cP_{T_\ell}^i(\epsilon)$} \label{qp-sss1} 
  \State $T^i\leftarrow T$
\Else
 \Comment{{\em Find a subset of $S$  that belongs to $\cP^i_{T_{\ell}}(\epsilon)$}}
     \If{$|{T}|\ge\frac{1}{\epsilon}$} \label{qp-sss2}
         \State $T\leftarrow $ the set of the first $\frac{1}{\epsilon}$ elements added to $T$ 
         \State $h\leftarrow1$; $V_1\leftarrow S\backslash {T}$\label{qp-sss2-}
      \Else
	 \State $T^i\leftarrow{T}$
        \State Find a partition $V_1,\ldots, V_h$ of $S\backslash {T}$ s.t. $\exists j\in[r_i]$: $\sum_{k\in V_p}q_k^{ij}\geq\frac{\epsilon}{2r_i} w^{ij}_{T}~~\forall p\in[h]$ \label{qp-sss3}
     \EndIf
	\State $K\leftarrow\{k\in {T}\,|\,f(S\backslash\{k\})\ge (1-\epsilon)f(S)\}$
	\State $P\leftarrow\{p\in [h]\,|\,f(S\backslash V_p)\ge (1-\epsilon)f(S)\}$\label{qp-sss3-} 
	\If{$K\not=\varnothing$}
          \State Pick arbitrary $\hat{k}\in K$
          \State $S\leftarrow S\backslash\{\hat{k}\}$\label{qp-sss4}
      \Else
          \State Pick arbitrary $\hat{p}\in P$
          \State $S\leftarrow S\backslash V_{\hat{p}}$\label{qp-sss5}
      \EndIf
\EndIf
\EndFor
\State \Return $(\{T^1,\ldots,T^m\},S)$
\end{algorithmic}
\end{algorithm}
}

\begin{lemma}\label{qp-lem-pack}
	Consider sets $T \subseteq S\subseteq[n]$ such that~~~
(C1)
$\sum_{k\in S}q_k^i\in\cB(r_i,C_i)\setminus\cP^i_T(\epsilon)$; ~~~  
(C2)
$q_k^{ij}\le \frac{\epsilon}{2r_i} w^{ij}_T$, for all $j\in[r_i]$ and $k \in S \backslash T$. 
 Then there exist $j\in[r_i]$, $h\in[\frac{1}{\epsilon}-1,\frac{2r_i}{\epsilon})$, and a partition $\{V_1,\ldots, V_h\}$ of $S \backslash T$ such that $\sum_{k\in V_s}q_k^{ij}\geq\frac{\epsilon}{2r_i} w^{ij}_T$ for all $s\in [h]$.
\end{lemma}
\fullonly{
\begin{proof}
We define $\eta \triangleq \sum_{k \in T} q_k^i$ and $\kappa \triangleq\sum_{k \in S \backslash T}q^i_k$. Since $\eta+\kappa=\sum_{k\in S}q_k^i\in\cB(r_i,C_i)\setminus\cP^i_T(\epsilon)$, we claim that there is a $j\in[r_i]$ such that $\kappa^{j}>\frac{w^{ij}_T}{r_i}$. Indeed, consider the $(r_i-1)$-dimensional simplex whose $j$th vertex, for $j\in[r_i]$, is $(\eta^1,\ldots,\eta^{j-1},\eta^j+w^{ij}_T,\eta^{j+1},\ldots,\eta^{r_i})$. Then the point $\xi\in\RR^{r_i}$, defined by $\xi^j=\eta^j+\frac{w^{ij}_T}{r_i}$ for $j\in[r_i]$, lies on the simplex, implying that the whole box $\cC:=\{\nu\in\RR^{r_i}:~\eta\le \nu \le\xi\}$ lies on the same side, including $\eta$, of the hyperplane $H$ defined by the simplex. On the other hand, any facet of $\cP_T^i(\epsilon)$ is defined by the convex hull of a set of points on the ball, which lie on the other side of $H$. Consequently, the point $\eta+\kappa$ strictly lies on this other side, and hence outside the box $\cC$, implying the claim. 

Let us fix $j\in[r_i]$ as in the claim, and pack consecutive vectors $q_k^i,$ $k \in S \backslash T$, into batches such that the sum in each batch has the $j$th component of length in the interval $(\frac{\epsilon}{2r_i} w^{ij}_T,\frac{\epsilon}{r_i} w^{ij}_T]$. More precisely, we fix an order on $S \backslash T := \{1, \ldots, \ell\}$, and find indices $1=k_1<k_2<\cdots <k_{h'}<k_{h'+1}=\ell+1$ such that
\begin{eqnarray}\label{ee1}
\sum_{k=k_l}^{k_{l+1}-1}q_k^{ij}&\leq&\frac{\epsilon}{r_i} w^{ij}_T,\text { for $l=1,\ldots,h'$}\label{eq1}\\ \mbox{and \ } \sum_{k=k_l}^{k_{l+1}}q_k^{ij}&>&\frac{\epsilon}{r_i} w^{ij}_T\text { for $l=1,\ldots,h'-1$}; \label{eq2}
\end{eqnarray}
the existence of such indices is guaranteed by (C2).
It follows from \raf{eq2} that $\sum_{k=k_l}^{k_{l+1}-1}q_k^{ij} >\frac{\epsilon}{2r_i} w^{ij}_T$ for $l=1,\ldots,h'-1,$ since $q_{k_{l+1}}^{ij}\le\frac{\epsilon}{2r_i} w^{ij}_T$. It also follows that $\frac{1}{\epsilon}\le h'<\frac{2r_i}{\epsilon}+1$, since summing \raf{eq1} for $\ell=1,\ldots, h'$ yields
\begin{eqnarray}
\frac{\epsilon}{r_i} h' w^{ij}_T \ge \sum_{l=1}^{h'}\sum_{k=k_l}^{k_{l+1}-1}q_k^{ij}=\sum_{k=1}^{\ell} q_k^{ij}=\kappa^j\ge \frac{w^{ij}_T}{r_i}; \label{eq:kappa}
\end{eqnarray}
the last inequality follows from our assumption.
Similarly, summing \raf{eq2} for $\ell=1,\ldots, h'-1$ yields
\begin{eqnarray}
(h'-1) \frac{\epsilon}{r_i} w^{ij}_T < \sum_{l=1}^{h'-1}\sum_{k=k_l}^{k_{l+1}}q_k^{ij}\le 2\sum_{k=1}^{r_i} q_k^{ij}=2\kappa^j\le 2 w^{ij}_T,
\end{eqnarray}
where the last inequality follows from (C1).
Setting $V_l\triangleq\{k_l,k_l+1,\ldots,k_{l+1}-1\}$, for $l=1,\ldots,h'-2$,  $V_{h'-1}=\{k_{h'-1},\ldots,r\}$, and $h\triangleq h'-1$  will satisfy the claim of the Lemma.
\end{proof}
}

\begin{lemma}
\label{sub-pro}
Let $f:2^{[n]}\rightarrow \RR_+$ be a nonnegative submodular function and let $S\subseteq [n]$ be a non-empty set. If $\{S_1,\ldots, S_k\}$ is a partition of $S$ into $k$ disjoint subsets, then 
$f(S\backslash S_i)\ge (1-\frac{1}{k})f(S)$ for some $i\in [k]$.
\end{lemma}
\fullonly{
\begin{proof}
We prove the claim by contradiction. Suppose that $f(S\backslash S_i)<(1-\frac{1}{k})f(S)$ holds for all $i\in [k]$. 
We claim by induction on $i=1,2,\ldots,k$ that 
\begin{equation}\label{sub-e1}
f(S\backslash\bigcup_{j=1}^iS_j)<\left(1-\dfrac{i}{k}\right)f(S),
\end{equation}
which, when applied with $i=k$, would give the contradiction $f(\varnothing) <0$.
The claim is true for $i=1$ by assumption. Let us assume it is true up to $i-1$.
By the submodularity
of $f$, we have:
\begin{align*}
f((S\backslash\bigcup_{j=1}^{i-1}S_j)\cup (S\backslash S_i))  + f((S\backslash \bigcup_{j=1}^{i-1}S_j)\cap (S\backslash S_i))& \leq f(S\backslash\bigcup_{j=1}^{i-1}S_j) + f(S\backslash S_i)\\
\end{align*}
implying by the induction hypothesis and the assumption that $f(S\backslash S_i)<(1-\frac{1}{k})f(S)$ that
\begin{align*}
f(S)+f(S\backslash\bigcup_{j=1}^{i}S_j)& < \left(1-\dfrac{i-1}{k}\right)f(S)+(1-\frac{1}{k})f(S),
\end{align*}
 proving the claim.
\end{proof}
}

The following results are straightforward consequences of 
Theorem \ref{QP-t2} and the results in \cite{kul-sha-tam:c:maximizing-submodular-set-function-subject-to-multiple-linear-constraint}.  

\begin{corollary}
There is a $(1-\epsilon)^m(\frac{1}{4} - \epsilon)$-approximation algorithm for the {\PS} problem, for any $\epsilon > 0$, when $m$ and each $r_i\triangleq\cpr(Q_i)$ are fixed. For the case when the submodular function is monotone, the approximation factor is $(1-\epsilon)^m(1-\frac{1}{e}-\epsilon)$, for any $\epsilon > 0$.  
\end{corollary}

\fullonly{
\begin{remark}\label{r2}
Note that a slight modification of Algorithm \ref{PSAPPROX} and Algorithm \ref{QC-Construct} yields a PTAS for the {\sc Pack-Lin} problem. To show that, we do some following changes. Firstly, in the step \ref{s1} of Algorithm \ref{PSAPPROX}, the solution $S$ will be computed by a PTAS for the $m$-dimensional knapsack problem instead of using the $\alpha(\epsilon)$-approximation algorithm for the submodular objective case. As a result, we have $u( \hat S ) \ge (1-\epsilon) u(S') \ge (1-\epsilon)^{m+1} \OPT$. Secondly, for Algorithm \ref{QC-Construct}, we do the same steps 1-25. Then, the solution $S^i$ will be obtained from $S^{i-1}$ by dropping the smallest utility-demand or group of demands with large component,
ensuring that the set $S^i$ of remaining demands can be packed inside $\Pc^i_{
T^i}(\epsilon)$. Hence, in the former case, $u(S^i) \ge (1 -\epsilon)u(S^{i-1})$,  and in the latter case, $u(S^i) \ge (1 - \frac{1}{h})S^{i-1} \ge \frac{1-2\epsilon}{1-\epsilon}S^{i-1} \ge (1 - 2\epsilon)S^{i-1}$.
\end{remark}
}


%% file: minimization.tex
\subsection{Minimization Problem $-$ a Greedy-based Approach}\label{sec:min}

In this section we consider problem \textsc{Cover-lin} with one quadratic constraint. We want to minimize a linear function $f(x)=u^Tx$, subject to $x^TQx\ge C^2$ and $x\in\{0,1\}^n$, where $Q\in\mathcal{CP}^*_n$ and $u\in\RR_+^n$. 
Note that the convex programming-based method can not be applied here since the relaxed problem, which is obtained by considering $x_i\in [0,1]$ for all $i\in[n]$, is {\it non-convex}, and thus we do not know if it can be solved efficiently. Furthermore, one can easily show that this programming relaxation has a bad integrality gap (of at least $C^2$), and thus is not a good choice for approximation.  Instead, we will follow a geometric approach. 

Let $I=(u,Q,C)$ be an instance of the problem \CL, and $\epsilon>0$ be a fixed constant. We will construct a quasi-polynomial-time algorithm which produces an $(1+\epsilon)$-approximate solution to the instance $I$.

By guessing, we may assume that $B\le\OPT<(1+\epsilon)B$, where $B:=(1+\epsilon)^i\min_k u_k$ for some $i\in\ZZ_+$ (the number of possible guesses is $O(\log_{1+\epsilon} (n\cdot\max_k u_k/\min_k u_k))$). Let $T:=\{k~|~u_k\ge (1+\epsilon)B\}$ and $V:=\{k~|~u_k<\frac{\epsilon}{n} \cdot B\}$; we set $x_k=1$ for all $k\in T$, and $x_k=0$ for all $k\in V$, and assume therefore that we need to optimize over a set $N:=[n]\setminus(T\cup V)$, for which $u_k\in[\frac{\epsilon}{n}\cdot B,(1+\epsilon)B)$ for $k\in N$. Note that such restriction increases the cost of the solution obtained by at most $\epsilon\cdot\OPT$. 

As before, write $Q=UU^T$, for some $U\in\QQ^{n\times r}_+$, where $r=\cpr(Q)$, and for $k\in[n]$ define the vector $q_k\in\QQ_+^{r}$ to be the $k$th column of $U^T$. Define the conic region $\cR_T$ as in \raf{region} (with the index $i$ dropped). Then the problem now amounts to finding $S\subseteq N$ s.t. $q_S:=\sum_{k\in S}q_k$ is not in the interior of $\cR_T$.

We begin by partitioning the set of vectors (indices) in $N$ into $h:=r\left(\frac{\sqrt{r}}{\epsilon}\right)^{r-1}$ {\it space-classes} $N^1,\ldots,N^h$, with the following property: for all $s\in[h]$, there exists $\xi(s)\in\RR_+^r$ such that for all $k\in N^s$, it holds
\begin{equation}\label{angle-cond}
\frac{q_k^T\xi(s)}{\|q_k\|_2\|\xi(s)\|_2}\ge 1-\epsilon.
\end{equation}
Condition \raf{angle-cond} says that there is a fixed direction $\xi(s)$ such that the angle that any vector $q_k$, $k\in N^s$, makes with $\xi(s)$ is sufficiently small. 

We will rely on the following geometric facts in our analysis of the algorithm.
\begin{fact}\label{geo-fact1}
Let $a,b,\xi\in\RR^r_+$ be such that $\frac{a^T\xi}{\|a\|_2\|\xi\|_2}\ge 1-\epsilon$ and $\frac{b^T\xi}{\|b\|_2\|\xi\|_2}\ge 1-\epsilon$. Then $\frac{(a+b)^T\xi}{\|a+b\|_2\|\xi\|_2}\ge 1-\epsilon$.
\end{fact}
\fullonly{
\begin{proof}
Indeed, 
\begin{eqnarray*}
\frac{(a+b)^T\xi}{\|a+b\|_2\|\xi\|_2}&=& \frac{\|a\|_2}{\|a+b\|_2}\cdot\frac{a^T\xi}{\|a\|_2\|\xi\|_2}+\frac{\|b\|_2}{\|a+b\|_2}\cdot\frac{b^T\xi}{\|b\|_2\|\xi(s)\|_2}\\
&\ge&\frac{\|a\|_2+\|b\|_2}{\|a+b\|_2}(1-\epsilon)\ge1-\epsilon,
\end{eqnarray*}
by the triangular inequality.
\end{proof}
}
Note that Fact~\ref{geo-fact1} implies that, for any $S\subseteq N^s$, $q_S:=\sum_{k\in S}q_k$ also satisfies the condition \raf{angle-cond}. 

\begin{fact}\label{geo-fact2}
Let $a,b,\xi\in\RR^r_+$ be such that $\frac{a^T\xi}{\|a\|_2\|\xi\|_2}\ge 1-\epsilon$ and $\frac{b^T\xi}{\|b\|_2\|\xi\|_2}\ge 1-\epsilon$. Then $\frac{a^Tb}{\|a\|_2\|b\|_2}\ge 1-5\epsilon$. 
\end{fact}
\fullonly{
\begin{proof}
Let $\bar a:=\frac{a}{\|a\|_2}$, $\bar b:=\frac{b}{\|b\|_2}$, and $\bar \xi:=\frac{\xi}{\|\xi\|_2}$. If $a,b,\xi$ all lie in the same subspace, then the claim follows since the angle between $a$ and $b$ is no more than the sum of the angles between  $a$ and $\xi$, and $b$ and $\xi$, which is at most is $\cos^{-1}((1-\epsilon)^2-\epsilon(2-\epsilon))\le \cos^{-1}(1-4\epsilon)$. Otherwise, let $\bar b=\widehat b+\tilde b$ be the orthogonal decomposition of $\bar b$ with respect to the 2-dimensional subspace formed by the two vectors $\bar a$ and $\bar \xi$, where $\widehat b$ is the projection of $\bar b$ into this space, and $\tilde b$ is the orthogonal component. Then $\frac{\widehat b^T\bar\xi}{\|\widehat b\|_2}=\frac{\bar b^T\bar\xi}{\|\widehat b\|_2}\ge \frac{1-\epsilon}{\|\widehat b\|_2}\ge 1-\epsilon$ (since $\|\widehat b\|_2\le \|\bar b\|_2=1$), which also implies that $\|\widehat b\|_2\ge 1-\epsilon$. Since $a$, $\widehat b$, and $\xi$ lie in the same subspace, it follows by the above argument that $\frac{\widehat b^T\bar a}{\|\widehat b\|_2}\ge (1-4\epsilon)$, and hence, $\bar b^T\bar a=\widehat b^T\bar a\ge (1-4\epsilon)\|\widehat b\|_2\ge(1-4\epsilon)(1-\epsilon)$, implying the claim.
\end{proof}   
}
\begin{fact}\label{geo-fact3}
Let $a,b\in\RR$ be such that $\frac{a^Tb}{\|a\|_2\|b\|_2}\ge 1-5\epsilon$ and $\|\widehat b\|_2=\lambda\|a\|_2$, where $\lambda\geq 1$ and $\widehat b:={\tt Pj}_{a}(b)$ is the projection of $b$ on $a$. Then for any vector $\eta\in\RR^r$, it holds that 
$$\|{\tt Pj}_{\eta}(b)\|_2\ge \lambda\left(\|{\tt Pj}_{\eta}(a)\|_2-\frac{\sqrt{5\epsilon(2-5\epsilon)}}{1-5\epsilon}\right)\|a\|_2.$$  
\end{fact}
\fullonly{
\begin{proof}
Since the statement is invariant under rotation, we may assume w.l.o.g. that $\eta=\bone_j$, the $j$th-dimensional unit vector in $\RR^r$. 
Write $b:=\widehat b+\tilde b$, where $\tilde b$ is the vector orthogonal to $a$ in the subspace spanned by $a$ and $b$. Then $\|{\tt Pj}_{\eta}(b)\|_2=b^j$ is the $j$th component of $b$, and $\|{\tt Pj}_{\eta}(\widehat b)\|_2=\widehat b^j$. Since
$$
\|\widehat b\|_2\ge \frac{a^T\widehat b}{\|a\|_2}=\frac{a^Tb}{\|a\|_2}\ge (1-5\epsilon)\|b\|_2,
$$
it follows that 
$$\|\tilde b\|_2= \sqrt{\|b\|_2^2-\|\widehat b\|_2^2}\le \sqrt{5\epsilon(2-5\epsilon)}\|b\|_2\le\frac{\sqrt{5\epsilon(2-5\epsilon)}}{1-5\epsilon}\|\widehat b\|_2=\frac{\lambda\sqrt{5\epsilon(2-5\epsilon)}}{1-5\epsilon}\|a\|_2.$$ Since 
$|b^j-\widehat b^j|=|\tilde b^j|\le\|\tilde b\|_2$, and $\widehat b=\lambda\|a\|_2 \cdot a$, the claim follows.    
\end{proof}
}
Condition \raf{angle-cond} together with Facts~\ref{geo-fact1} and~\ref{geo-fact2}  imply that for any two sets $S,S'\subseteq N^s$, we have $\frac{q_S^Tq_{S'}}{\|q_S\|_2\|q_{S'}\|_2}\ge 1-5\epsilon$.

The space partitioning can be done as follows. Let $q_T:=\sum_{k\in T}q_k$. We  partition the region $\cR_T$ into disjoint regions $\cR_T(1),\ldots,\cR_T(h)$, obtained as follows. Let $\mu=\bar w_T\cdot\bone$, where $\bar w_T:=\max_{j\in[r]}w_T^j$ (recall the definition of $w_T^{j}$ from \raf{sqrt}), and define the $r$-dimensional box $\cC_T:=\{\nu\in\RR_+^r~|~q_T\le \nu\le \mu\}$. Note that $\cR_T\subseteq\cC_T$. We grid the $r$ facets of $\cC_T$ that do not contain the point $q_T$ by interlacing equidistant $(r-2)$-dimensional parallel hyperplanes with inter-separation $\frac{\epsilon \bar w_T}{\sqrt r}$, for each $j\in[r]$, and with the $j$th principal axis as their normal. Note that the total number of grid cells obtained is $h$; let us call them $\cC_1,\ldots,\cC_h$ (these are $(r-1)$-dimensional  hypercubes). We then define the region $\cR(s)$ as the $r$-dimensional simplex $\cR(s):=\conv(\{q_T\}\cup\cC_s)$ and $\xi(s)=c(s)-q_T$ as the designated vector, where $c(s)$ is the vertex center of the cell $\cC_s$. Note that the angle condition~\raf{angle-cond} is satisfied. Indeed, consider any vector $q_k$ such that $q_T+q_k\in\cR_T(s)$. Let the ray $\{q_T+\lambda q_k:\lambda\ge 0\}$ hit the boundary cell $\cC_s$ in the point $x$. Consider the triangle formed by the three points $q_T$, $c(s)$ and $x$. Then, by construction, the distances between $q_T$ and both $c(s)$ and $x$ are at most $\bar w$, whereas the distance between $c(s)$ and $x$ is at most $\epsilon\bar w\sqrt{r}$. It follows that the angle between the two vectors $q_k$ and $\xi(s)$ is no more than $\sin^{-1}\epsilon$, implying \raf{angle-cond}.

Finally, we define $N^s:=\{k~|~q_T+q_k\in\cR_T(s)\}$. This gives the required space partitioning of the vectors.

Next, we group the set of items in $N$ into $\ell:=1+\log_{1+\epsilon}\frac{n}{\epsilon}$ (some possibly empty) {\it utility-classes} $N_1,\ldots,N_\ell$, where $N_l=\{k~|~u_k\in[\frac{\epsilon}{n} \cdot B(1+\epsilon)^{l-1},\frac{\epsilon}{n} \cdot B(1+\epsilon)^{l})\}.$ Note that for all $k,k'\in N_l$, we have
\begin{equation}\label{utility-cond}
 u_k\le u_{k'}(1+\epsilon).
\end{equation}

We can show that for a set of vectors $\{q_k:~k\in N^s\cap N_l\}$ that lie in the same region and same utility-group, the {\it greedy} algorithm that processes the vectors in {\it non-increasing} order of length gives an $O(\epsilon)$-optimal solution. For simplicity we assume first that $T=\varnothing$ and $q_T=\bzero$.  

\begin{algorithm}[!htb]
	\caption{ \GRC$(u,\{q_k\}_{k\in N},C)$} \label{GRC}
\begin{algorithmic}[1]
\Require A cost vector $u\in\RR_+^N$; accuracy parameter $\epsilon$; vectors $q_k\in\QQ_+^r$ satisfying \raf{utility-cond} and \raf{angle-cond}; a demand $C\in\RR_+$;  
\Ensure $9\epsilon$-optimal solution $S$ to \CL
\State $S':=\min\{u(S)~|~S\subseteq N, ~|S|\le\frac{1}{\epsilon},~q_S \text{ is feasible}\}$ \label{g-s00}
\State Order the vectors $q_k$, $k\in N$, such that $\|q_1\|_2\ge \|q_2\|_2\ge\cdots$\label{g-s0}
\State $S \leftarrow \varnothing$; $k\leftarrow 0$
\While{$\|\sum_{k\in S}q_k\|_2<C$}\label{g-s1}
  \State $k\leftarrow k+1$
  \State $S\leftarrow S\cup\{k\}$\label{g-s2}
\EndWhile
\If{$u(S)\le u(S')$}
\State \Return $S$
\Else
\State \Return $S'$
\EndIf
\end{algorithmic}
\end{algorithm}

\begin{lemma}\label{greedy main}
Consider instance of problem \CL\ described  by a set of vectors $\{q_k\}_{k\in N}$ satisfying \raf{angle-cond} and \raf{utility-cond}. Then, for any sufficiently small constant $\epsilon>0$, Algorithm~\ref{GRC} outputs a solution $S$ satisfying $u(S)\le (1+9\epsilon)\OPT$. 
\end{lemma}
\fullonly{
\begin{proof}
Let 
$S^*
$ be an optimal solution. Since we consider every possible feasible solution of size at most $\frac{1}{\epsilon}$ in step~\ref{g-s00}, we may assume w.l.o.g. that $|S^*|\ge\frac{1}{\epsilon}$.

We claim that $|S|\le\frac{|S^*|}{1-5\epsilon}+1$. To see this claim, let $q_{S^*}:=\sum_{k\in S^*}q_k$, and for $k\in N$, denote by $\widehat q_k:={\tt Pj}_{q_{S^*}}(q_k)$ the projection of $q_k$ on $q_{S^*}$. Let $q_w\in S$ be the last item added to $S$ in step~\ref{g-s2} of the algorithm. Then it is clear that
\begin{equation}\label{GRC-proj}
\sum_{k\in S\backslash\{w\}}\|\widehat q_k\|_2<\|q_{S^*}\|_2,
\end{equation}
since otherwise $S':=S\backslash\{w\}$ would satisfy the condition in line~\ref{g-s1} (as $\|\sum_{k\in S'}q_k\|_2\ge \sum_{k\in S'}\|\widehat q_k\|_2\ge\|q_{S^*}\|_2\ge C$, by the feasibility of $S^*$). By the angle condition \raf{angle-cond}: $\|\widehat q_k\|_2\ge (1-5\epsilon)\|q_k\|_2$ for all $k\le N.$
It follows by \raf{GRC-proj} that
\begin{eqnarray*}
(1-5\epsilon)|S'|\min_{k\in S'}\|q_k\|_2&\le&(1-5\epsilon)\sum_{k\in S'}\|q_k\|_2\le\sum_{k\in S'}\|\widehat q_k\|_2\\&<&\|q_{S^*}\|_2=\sum_{k\in S^*}\|\widehat q_k\|_2\le \sum_{k\in S^*}\|q_k\|_2\le |S^*|\max_{k\in S^*}\|q_k\|_2.
\end{eqnarray*}
The claim follows since $\max_{k\in S^*}\|q_k\|_2\le\min_{k\in S'}\|q_k\|_2$, by the greedy order in line~\ref{g-s0}.

By the utility condition~\raf{utility-cond},  
\begin{eqnarray*}
u(S)\le |S|\max_k u_k\le|S|(1+\epsilon)\min_ku_k\le\left(\frac{|S^*|}{1-5\epsilon}+1\right)(1+\epsilon)\min_ku_k\le\left(\frac{1}{1-5\epsilon}+\epsilon\right)(1+\epsilon)u(S^*).
\end{eqnarray*}
The lemma follows.
\end{proof}
}
Our QPTAS is given as Algorithm~\ref{cl-qptas}. For simplicity, we assume that the algorithm has already a correct guess of the bound $B$ on the value of the optimal solution. After decomposing the instance into classes according to utility and region (steps~\ref{c-s0} and \ref{c-s01}), the algorithm enumerates over all possible selections of a nonnegative integer $n_{s,l}$ associated to each region $s$ and a utility class $l$ (step~\ref{c-s02}); this number $n_{s,l}$ represents the largest length 
vectors that are taken in the potential solution from the set $N^s\cap N_l$. However, for technical reasons, the algorithm does this only for pairs $(s,l)$ for which the set $N^s\cap N_l$ contributes at least $\frac{1}{\epsilon}$ in the optimal solution; the set of pairs that potentially do not satisfy this can be identified by enumeration (steps~\ref{c-s001} and \ref{c-s002}). 

\begin{algorithm}[!htb]
	\caption{ \CLQPTAS$(u,Q,C)$} \label{cl-qptas}
\begin{algorithmic}[1]
\Require A cost vector $u\in\RR_+^N$; accuracy parameter $\epsilon$; matrix $Q\in\mathcal{CP}_n^*$; a demand $C\in\RR_+$  
\Ensure An $O(\sqrt{\epsilon r})$-optimal solution $S$ to \CL \Comment{{\em $r=\text{cp-rank}(Q)$}}
\State Obtain sets $T$ and $V$ as explained above and set $N\leftarrow[n]\backslash(T\cup V)$
\State Decompose $Q$ as $Q=UU^T$, where $U\in\QQ^{n\times r}_+$ \label{c-s0}
\State Let $\{q_k\}_{k\in N}$ be the columns of $U^T$ corresponding to the indices in $N$
\State Decompose the set $N$ into space-classes $N^1,\ldots,N^h$ and utility classes $N_1,\ldots,N_\ell$\label{c-s01} 
\State $S\leftarrow T \cup N$ \Comment{Assume instance is feasible}
\For{each subset of pairs $\cG\subseteq [h]\times[\ell]$} \label{c-s001}
  \For{each possible selection $(T_{s,l}\subseteq N^s\cap N_l:~|T_{s,l}|\le\frac{1}{\epsilon},~(s,l)\in\cG)$}\label{c-s002}
     \For{each possible selection $(n_{s,l}\in\{1,\ldots,|N^s\cap N_l|\}:~s\in([h]\times[\ell])\setminus \cG)$}\label{c-s02}
        \State Let $S_{s,l}$ be the set of the $n_{s,l}$ vectors  with largest length in $N^s\cap N_l$   
        \State $S' \leftarrow T \cup \left(\bigcup_{(s,l)\in\cG}T_{s,l}\right)\cup\left(\bigcup_{s,l}S_{s,l}\right)$
        \If{$\|\sum_{k\in S'}q_k\|_2\ge C$ and $u(S')<u(S)$}\label{c-s1}
            \State $S\leftarrow S'$\label{c-s2}
        \EndIf
     \EndFor
   \EndFor
\EndFor
\State \Return $S$
\end{algorithmic}
\end{algorithm}

\begin{lemma}\label{c-l1}
For any sufficiently small $\epsilon>0$ (for instance, $\epsilon<\frac{1}{(r+1)^4}$ for $r\ge 2$), Algorithm~\ref{cl-qptas} runs in time $n^{O((\frac{\sqrt{r}}{\epsilon})^{r+1}\log n)}$ and outputs a solution $S$ satisfying $u(S)\le (1+O(\sqrt{\epsilon r}))\OPT$ for any instance of \CL. 
\end{lemma}
\fullonly{
\begin{proof}
The running time is obvious since it is dominated by the number of selections in steps~\ref{c-s001}, \ref{c-s002} and~\ref{c-s02}, which is at most $\left(\frac{n}{h\ell}\right)^{O(\frac{h\ell}{\epsilon})}=n^{O((\frac{\sqrt{r}}{\epsilon})^{r+1}\log n)}$. 

To see the approximation ratio, fix 
\begin{equation}\label{lambda}
\lambda:=\left(1-\frac{\sqrt{5\epsilon(2-5\epsilon)r}}{1-5\epsilon}\right)^{-1}.
\end{equation}
Let $S^*$ be an optimal solution (within $N$), and for $s\in[h]$ and $l\in[\ell]$, let $S_{s,l}^*:=S^*\cap N^s\cap N_l$. Let $\cG:=\{(s,l)\in[h]\times[\ell]:~|S^*_{s,l}|<\frac{1}{\epsilon}\}$ and $T_{s,l}:=S^*\cap N^s\cap N_l$, for $(s,l)\in\cG$.

Let $\widehat S:=S^*\backslash\left(\bigcup_{(s,l)\in\cG}(N^s\cap N_l)\right)$.
For $k\in N^s\cap N_l$, let $\widehat q_k:={\tt Pj}_{q_{S^*_{s,l}}}(q_k)$ be the projection of $q_k$ on $q_{\widehat S_{s,l}}:=\sum_{k\in \widehat S_{s,l}}q_k$. Define the set of pairs 
$$\cH:=\left\{(s,l)\in [h]\times[\ell]:~\|\sum_{k\in N^s\cap N_l}\widehat q_k\|_2\ge\lambda \|q_{\widehat S_{s,l}}\|_2 \text{ and }|S^*_{s,l}|\ge\frac{1}{\epsilon}\right\}.$$

Then according to Lemma~\ref{greedy main} (or more precisely, its proof), for every $(s,l)\in \cH$ there a choice $n_{s,l}\in \{1,\ldots,|N^s\cap N_l|\}$, such that if $S_{s,l}$ is the set of the $n_{s,l}$ vectors  with largest length in $N^s\cap N_l$, then 
$\sum_{k\in S_{s,l}} \|\widehat q_k\|_2\ge\lambda\| q_{\widehat S_{s,l}}\|_2$ and $u(S_{s,l})\le \left(\frac{\lambda}{1-5\epsilon}+\epsilon\right)(1+\epsilon)u(\widehat S_{s,l})$. By this, we can define $p_{k}=\tau \cdot q_k$, for $\tau:=\frac{\lambda\| q_{\widehat S_{s,l}}\|_2}{\sum_{k\in S_{s,l}} \|\widehat q_k\|_2}\le 1$, such that $\sum_{k\in S_{s,l}} \|\widehat p_k\|_2=\lambda\|q_{\widehat S_{s,l}}\|_2$, where $\widehat p_k:={\tt Pj}_{q_{\widehat S_{s,l}}}(p_k)$.

Let us now apply Fact~\ref{geo-fact3} with $a=a_{s,l}:=q_{\widehat S_{s,l}}$, $b=b_{s,l}:=p_{S_{s,l}}$, and $\eta=q_{\widehat S}$ to get that
\begin{equation}\label{gc-e1}
\|{\tt Pj}_{\eta}(b_{s,l})\|_2\ge\lambda\left(\|{\tt Pj}_{\eta}(a_{s,l})\|_2- \frac{\sqrt{5\epsilon(2-5\epsilon)}}{1-5\epsilon}\right)\|a_{s,l}\|_2. 
\end{equation}
Summing the above inequalities for all $(s,l)\in\cH$, we get
\begin{eqnarray}\label{gc-e2}
\sum_{(s,l)\in\cH}\|{\tt Pj}_{\eta}(b_{s,l})\|_2&\ge&\lambda\left(\sum_{(s,l)\in\cH}\|{\tt Pj}_{\eta}(a_{s,l})\|_2- \frac{\sqrt{5\epsilon(2-5\epsilon)}}{1-5\epsilon}\sum_{(s,l)\in\cH}\|a_{s,l}\|_2\right).
\end{eqnarray}
Note that $\sum_{(s,l)\in\cH}a_{s,l}=\eta$ by construction. By the Cauchy-Schwarz inequality and the nonnegativity of the vectors, 
\begin{eqnarray*}\label{gc-e3}
\|\sum_{(s,l)\in\cH}{\tt Pj}_{\eta}(a_{s,l})\|_2&=&\|\eta\|_2=\|\sum_{(s,l)\in\cH}a_{s,l}\|_2\ge\frac{1}{\sqrt{r}}\|\sum_{(s,l)\in\cH}a_{s,l}\|_1=\frac{1}{\sqrt{r}}\sum_{(s,l)\in\cH}\|a_{s,l}\|_1\ge\frac{1}{\sqrt{r}}\sum_{(s,l)\in\cH}\|a_{s,l}\|_2.
\end{eqnarray*} 
Using this in \raf{gc-e2}, we obtain
\begin{eqnarray}\label{gc-e4}
\sum_{(s,l)\in\cH}\|{\tt Pj}_{\eta}(b_{s,l})\|_2&\ge&\lambda\left(1- \frac{\sqrt{5\epsilon(2-5\epsilon)r}}{1-5\epsilon}\right)\sum_{(s,l)\in\cH}\|{\tt Pj}_{\eta}(a_{s,l})\|_2\ge \sum_{(s,l)\in\cH}\|{\tt Pj}_{\eta}(a_{s,l})\|_2,
\end{eqnarray} 
by our choice of $\lambda$. From \raf{gc-e2} and $q_{S_{s,l}}\ge b_{s,l}$, it follows by the feasibility of $S^*$ that the solution defined by $S=T \cup \left(\bigcup_{(s,l)\in\cG}T_{s,l}\right)\cup\left(\bigcup_{s,l}S_{s,l}\right)$ is feasible.  

Clearly, one of the choices in each of the enumeration steps~\ref{c-s001},~\ref{c-s002}, and~\ref{c-s02} will capture the above choices $\cG$, $T_{s,l}$ for $(s,l)\in\cG$, and $n_{s,l}$, for $(s,l)\in([h]\times[\ell])\backslash\cG$. It follows the procedure returns a solution $S$ with utility:
\begin{eqnarray*}
u(S)&\le&u(T)+\sum_{(s,l)\in\cG}u(T_{s,l})+\sum_{(s,l)\not\in\cG}u(S_{s,l})\\
&\leq& u(T)+\sum_{(s,l)\in\cG}u(T_{s,l})+\left(\frac{\lambda}{1-5\epsilon}+\epsilon\right)(1+\epsilon)\sum_{(s,l)\not\in\cG}u(\widehat S_{s,l})\\
&\le&\left(\frac{1}{1-5\epsilon-\sqrt{5\epsilon(2-5\epsilon)r}}+\epsilon\right)(1+\epsilon)\left(u(T)+\sum_{(s,l)\in\cG}u(T_{s,l})+\sum_{(s,l)\not\in\cG}u(\widehat S_{s,l})\right)\\
&\le& (1+O(\sqrt{\epsilon r}))u(S^*).
\end{eqnarray*}
The lemma follows.
\end{proof}
}
\fullonly{
\subsection{Dealing with the Approximation in the Decomposition}\label{err} 
According to Corollary~\ref{cp-rank-approx}, given a matrix $Q_i$ of fixed cp-rank $r_i$, a decomposition of the form $Q_i:=U_iU_i^T-\Delta_i$, where $U_i\in\QQ^{r_i\times n}$, $\Delta_i\in\RR_+^{n\times n}$, and $\|\Delta_i\|_\infty\le \delta$, can be done in time $\poly(\cL,n^{O(r_i^2)},\log\frac{1}{\delta})$ for any $\delta>0$. This leads to a technical issue: if one uses the approximate decomposition $Q_i:=U_iU_i^T$ to solve the BQC problem, then the resulting solution can be either infeasible or far from optimal. When $m=O(1)$, this issue can be dealt with at a small loss in the approximation ratio as follows. 
Let $Q_i^{kj}$ denote the $(k,j)$th entry of $Q_i$. Note that, if the $k$th diagonal element $Q_i^{kk}$ of $Q_i$ is $0$, then every entry in the $k$th row and the $k$th column of $Q_i$ is $0$, since $Q_i$ is positive semi-definite.  

We consider the case when $f$ is linear.
Let $S^*$ be an optimal solution for \PQC\ (resp., \CQC) and $x^*$ be its characteristic vector. The idea is to show that there is a feasible solution $\widetilde S\subseteq S^*$ (resp., $\widetilde S\supseteq S^*$) such that $f(\widetilde S)\ge (1-\epsilon)f(S^*)$ (resp., $f(\widetilde S)\le (1+\epsilon)f(S^*)$), and such that $\widetilde x^TQ_i \widetilde x\le \widetilde C_i$ (resp., $\widetilde x^TQ_i \widetilde x\ge \widetilde C_i$), where $\widetilde x$ is the characteristic vector of $\widetilde S$ and $\widetilde C_i:=C_i-\delta n^2$ (resp., $\widetilde C_i:=C_i+\delta n^2$).
Then it would be enough to solve a new instance with the same objective function but with the constraints $x^TU_iU_i^Tx\le  C_i$ (resp., $x^TU_iU_i^Tx\ge \widetilde C_i$).
Note that $\widetilde x$ is feasible for the new instance, since $\widetilde x^TU_iU_i^T \widetilde x=\widetilde x^TQ_i \widetilde x+\widetilde x^T\Delta_i \widetilde x\le \widetilde C_i+n^2\delta=C_i$ (resp., $\widetilde x^TU_iU_i^T \widetilde x=\widetilde x^TQ_i \widetilde x+\widetilde x^T\Delta_i \widetilde x\ge \widetilde C_i$).

 Suppose that $x$ is an $\alpha$-approximate solution for the new instance of \PQC, then $x^TQ_ix=x^TU_iU_i^Tx-x^T\Delta_ix\le C_i$, since $\Delta_i\ge 0$, and hence $x$ is an $\alpha(1-\epsilon)$-approximation for the original instance. Similarly, if $x$ is an $\alpha$-approximate solution for the new instance of \CQC, then $x^TQ_ix=x^TU_iU_i^Tx-x^T\Delta_ix\ge \widetilde C_i-\delta n^2=C_i$, and hence $x$ is an $\alpha(1+\epsilon)$-approximation for the original instance. 

The $\epsilon$-approximate solution $\widetilde S$ can be obtained as follows.
Let $\epsilon>0$ be a fixed constant, and define $\lambda=\frac{m}{\epsilon}$. Let $\cU$ be the $\lambda$ items of highest utility in the optimal solution $S^*$. This gives a partition of $[n]$ into three sets $\cU\cup \cV\cup N$, as described in Section \ref{lin-convex}. We can make the following assumptions w.l.o.g.:

\begin{itemize}
\item[(A1)] $|S^*|\ge\frac{m}{\epsilon}$;
\item[(A2)] for each $i\in[m]$ the set $\kappa_i:=\{k\in S^*\cap N:~Q_i^{kk}>0\}\neq\varnothing$ (resp., $\kappa_i:=\{k\in N\setminus S^*:~Q_i^{kk}>0\}\neq\varnothing$).
\end{itemize}
(A1) can be assumed since otherwise the algorithm that enumerates over all possible sets of size at most $\frac{m}{\epsilon}$ is optimal. (A2) can be assumed since if it is violated by $i\in[m]$, then after fixing the set $\cU$ (resp, after fixing the set $\cU$ and fixing the variables in the set $\{k\in N:~Q_i^{kk}=0\}$ to $1$), the $i$th constraint becomes redundant. Thus the algorithm that enumerates all subsets $\cU$ of size at most $\frac{m}{\epsilon}$ and all subsets of constraints, removing a subset at a time (resp., enumerates all subsets $\cU$ of size at most $\frac{m}{\epsilon}$ and all subsets of constraints, removing a subset at a time and fixing the variables in the set $\{k\in N:~Q_i^{kk}=0\}$ to $1$), can assume an instance satisfying (A2).       

Then, assuming (A1) and (A2), for each $i\in[m]$, we pick an index $k(i)\in \kappa_i$ and set the corresponding variable to $0$ (resp., to $1$). This yields an $\epsilon$-optimal solution $\widetilde S$, whose characteristic vector $x$ also satisfies $x^TQ_ix\le (x^*)^TQ_ix^*-Q_i^{k(i),k(i)}\le \widetilde C_i$ (resp., $x^TQ_ix\ge (x^*)^TQ_ix^*+Q_i^{k(i),k(i)}\ge \widetilde C_i$), if we choose $\delta:=\frac{\min_{i,k}Q^{kk}_i}{n^2}$.   
}

%% file: application.tex
\section{Applications of The Approximation Scheme}
\label{sec:appl}
\subsection{A PTAS for Multi-Objective Optimization}
In this section we study optimization problem with multiple objectives. In such a problem it is concerned with optimizing more than one objective function simultaneously. In fact, we consider the following binary multi-objective quadratic problem (BMOQ): 
\begin{eqnarray*}
\textsc{(BMOQ)} \qquad & \displaystyle \{\max f(x), \min g(x)\} \label{mbq}\\
\text{subject to}\qquad & \displaystyle x\in X\subseteq\{0,1\}^n
\end{eqnarray*}
where $f(x)=\{f_i(x)\}_{i\in \mathcal{I}}$ and $g(x)=\{g_j(x)\}_{j\in\mathcal{J}}$, and 
\[
X=\{x\in\{0,1\}^n|~x^TQ_ix + q_i^Tx \le C_i^2;Ax\le b;i\in[m]\},
\]
where $Q_i\in\mathcal{CP}_n^*,q_i\in\RR_+^n$ for all $i\in[m]$, $A\in\RR_+^{d\times n}$, $b\in\RR_+^d$, and $m,d$ are constant.

Typically, a feasible solution that simultaneously optimizes each objective may not exist due to the trade-off between the different objectives. This issue can be captured by the notion of {\em Pareto-optimal frontier} (or {\em Pareto curve}), which is a set of all feasible solutions whose vector of the objectives is not {\em dominated} by any other one. Unfortunately, the size of such a Pareto-optimal frontier set is often exponential in size of the input. Hence, a natural goal in this setting is, given an instance of a (BMOQ) problem, to efficiently achieve an $\epsilon$-approximation of the Pareto optimality that consists of a number of solutions that is polynomial in size of the instance, for every constant $\epsilon>0$. Formally, let $I$ be an instance of the multi-objective optimization problem (BMOQ), consider the following definitions taken from \cite{mit-sch:j:a-general-framework-for-designing-approximation-scheme}:

\begin{definition}
A Pareto-optimal frontier $P$ of $I$ is a subset of $X$, such that for any solution $S\in P$ there is no solution $S'\in X$ such that $f_i(S')\ge f_i(S)$ for all $i\in \mathcal{I}$ and $g_j(S')\le g_j(S)$ for all $j\in \mathcal{J}$, with strict inequality for at least one of them.
\end{definition}  

\begin{definition}
For $\epsilon>0$, an $\epsilon$-approximate Pareto-optimal frontier of $I$, denoted by $P_{\epsilon}$, is a set of solutions, such that for all $S\in X$, there exists a solution $S'\in P_{\epsilon}$ such that $f_j(S')\ge (1-\epsilon)f_j(S)$ and $g_j(S')\le g_j(S)/(1-\epsilon)$ for all $i\in \mathcal{I},j\in \mathcal{J}$.
\end{definition}

Computing an (approximate) Pareto-optimal frontier has been known to be a very efficient approach in designing approximation schemes for combinatorial optimization problems with multiple objectives \cite{erl-kel-pre:j:approximating-multiobjective-knapsack,pap-yan:c:on-the-approximability-of-trade-off-and-optimal-access,mit-sch:j:a-general-framework-for-designing-approximation-scheme,mit-sch:j:an-fptas-optimizing-a-class-of-lower-rank-functions-over-a-polytope}. Erlebach et al. \cite{erl-kel-pre:j:approximating-multiobjective-knapsack} give a PTAS for multi-objective multi-dimensional knapsack problem. 
In this paper, we will extend their method to the case with quadratic constraints. Let $I$ be an instance of the BMOQ problem, and denote $f_j(x)=a_j^Tx,g_j(x)=b_j^Tx$ for all $j\in [p]$. Let $\epsilon>0$ be a fixed constant. The Algorithm {\sc Pareto-Opt} below will produce a set $P_{\epsilon}$, which is an $\epsilon$-approximate Pareto-optimal frontier to the instance $I$, in polynomial time. Moreover, the size of $P_{\epsilon}$ is guaranteed to be also polynomial in size of the input for every fixed $\epsilon>0$. 

We first compute $\epsilon^2$-approximate solutions $S_j,T_j$ to the maximization problem with single objective functions $f_j=a_j^Tx,g_j=b_j^Tx$ subject to $x\in X$, respectively, by using an extension of Algorithm \ref{Lin-QC-PTAS}
. As a result, for any feasible solution $S$, we have $f_j(S)\le (1-\epsilon^2)^{-1}f_j(S_j)$ and $g_j(S)\le (1-\epsilon^2)^{-1}g_j(T_j)$ for all $j\in[p]$. For each $j\in[p]$, we consider $v_j+2$ (lower) bounds $\{\alpha_{j\lambda_j},\, 0\le \lambda_j\le v_j+1\}$ for the objective $f_j$ and $w_j+3$ (upper) bounds $\{\beta_{j\eta_j},\, 0\le \eta_j\le w_j+2\}$ for the objective $g_j$. We define the set  $\Lambda\times\Gamma\subseteq \RR_+^{2p-1}$, which contains all the tuples $(\lambda_1,\ldots,\lambda_{p-1},\eta_1,\ldots,\eta_{p})$. Note that the size of $\Lambda\times\Gamma$ is bounded by a polynomial in size of the input since $p$ is constant and $v_j,w_j$ are also bounded by a polynomial in size of the input. The main idea of the algorithm is that, for each tuple of the bounds, we try to maximize the last objective $f_p$ subject to $f_j\ge \alpha_{j\lambda_j}$, for all $j\in[p-1]$, and $g_j\le \beta_{j\eta_j}$, for all $j\in[p]$. To do that, we consider all possibility of choosing a subset of $[n]$ of at most (a constant number) $\lambda$ of items. Denote $\Sigma$ as the set of all such subsets and let $\mathcal{X}=\Sigma^{p}$. For each tuple $(X_1,\ldots,X_p)\in \mathcal{X}$, where $X_j$ is a set of items of highest utility with respect to $f_i$, we construct a tuple $(Y_1,\ldots,Y_p)$ such that $Y_j$ is the set of all items that will not be put into the knapsack given that items in $X_j$  are the ones with highest utility with respect to $f_j$ in the solution. Now for each tuple $(X_1,\ldots, X_p,Y_1,\ldots, Y_p)$ such that $\bigcup_{j\in [p]} (X_j \cap Y_j)=\varnothing$, we solve the convex program (CP2$[\cU]$) and obtain (if exists) an $\epsilon$-optimal solution $x^*$:
\begin{align*}
\textsc{(CP2$[\cU]$)} \qquad& \displaystyle \max a_p^Tx \label{}\\
\text{subject to}\qquad & \displaystyle x^TQ_ix +q_i^Tx\le C_i^2, i\in[m],\\
\qquad & \displaystyle Ax \le b, \\
& \displaystyle a_j^Tx \ge \alpha_{j\lambda_j}, j\in[p-1],\\
& \displaystyle b_j^Tx \le \beta_{j\eta_j}, j\in[p],\\
& \displaystyle x_k=1, \text{ for } k\in \mathcal{U}, x_k=0, \text{ for } k\in \mathcal{V},\\
& \displaystyle x_k\in[0,1], \text{ for } k\in N.
\end{align*}
Define $t_i:= U_i^T[*;N]x^*_N$, $t_i':= \bone_\cU^TQ_i[\cU;N]x^*_N$, ${q_i^T}[*;N]x^*_N= v_i$ for all $i\in [m]$, $A[*;N]x^*_N= \sigma$, ${a_j^T}[*;N]x^*_N= \theta_j \text{ for } j\in[p-1]$, and  ${b_j^T}[*;N]x^*_N= \theta'_j \text{ for } j\in[p]$. We define a polytope $\cP(\cU)\subseteq[0,1]^{N}$ as follows:
\[\cP(\cU) = \left\{ \begin{gathered}
   \hfill \\
  y\in [0,1]^N \hfill \\
   \hfill \\ 
\end{gathered}  \right.\left| \begin{gathered}
  U_i^T[*;N]y \le t_i,~\bone_\cU^TQ_i[\cU;N]y\le t_i',{q_i^T}[*;N]y\le v_i, \text{ for }i\in[m] \hfill \\
   A[*;N]y\le \sigma\hfill \\
  a_j^T [*;N]y\ge \theta_j \text{ for } j\in[p-1],b_j^T [*;N]y\le \theta'_j \text{ for } j\in[p]\hfill \\ 
\end{gathered}  \right\}\]
We can find (if exists) a (BFS) $y$ in this polytope such that $a_p^Ty\ge a_p^Tx^*_N$. By rounding down this fractional solution $y$ and setting $x_k\in\{0,1\}$ according to the assumption $k\in\cU\cup\cV$, we obtain an integral solution $\overline x\in \{0,1\}^n$. 

There are totally $n^{p\lambda}$ tuples of the form $(X_1,\ldots, X_p,Y_1,\ldots, Y_p)$. For each such tuple such that the condition in Step \ref{check} is satisfied and (CP1$[\cU]$) is feasible, the algorithm outputs an integral solution $\overline x$. All possible integral solutions obtained in this way are collected in the set $P_{\epsilon}$. This completes the description of the algorithm.

\begin{algorithm}[!htb]
	\caption{ {\sc Pareto-Opt}$(\{a_j^T,b_j^T,Q_i,C_i\}_{j\in[p],~i\in[m]},\epsilon)$} \label{mul-PTAS}
\begin{algorithmic}[1]
\Require Utilities, demand vectors and capacities $(\{a_j^T,b_j^T,Q_i,C_i\}_{j\in[p],~i\in[m]})$; accuracy parameter $\epsilon$
\Ensure $(1-\epsilon)$-approximate Pareto optimality $P_{\epsilon}$
\State $P_{\epsilon}\leftarrow \varnothing$
\State Decompose $Q_i$ into $U_i^TU_i$, where $U_i$ has $r_i$ rows, for all $i\in [m]$\Comment{{\em $r_i=\text{cp-rank}(Q_i)$}}
\For{$j\in [p]$}
\State $S_j\leftarrow \text{PTAS}(a_j^T,Q_i,C_i,i\in[m],\epsilon^2)$ \Comment{{\em Find an $\epsilon^2$-approximate solution to $f_j$}}
\State $T_j\leftarrow \text{PTAS}(b_j^T,Q_i,C_i,i\in[m],\epsilon^2)$ 
\Comment{{\em Find an $\epsilon^2$-approximate solution to $g_j$}}
\State $v_j\leftarrow \left\lceil \log_{(1-\epsilon')^{-1}}f_j(S_j)\right\rceil$; $w_j\leftarrow \left\lceil \log_{(1-\epsilon')^{-1}}g_j(T_j)\right\rceil$
\label{step1}
\State $\alpha_{j0}\leftarrow 0$; $\alpha_{j\lambda_j}\leftarrow (1-\epsilon^2)^{1-\lambda_j}$ for all $\lambda_j\in[v_j+1]$
\State $\beta_{j0}\leftarrow 0$; $\beta_{j\eta_j}\leftarrow (1-\epsilon^2)^{1-\eta_j}$  for all $\eta_j\in[w_j+2]$
\EndFor
\State $\Lambda\leftarrow [0,v_1]\times \ldots \times [0,v_{p-1}]$; $\Gamma\leftarrow [0,w_1]\times \ldots \times [0,w_{p}]$
\State $\phi\leftarrow\sum_{i=1}^m(r_i+1)$; $\lambda\leftarrow\left\lceil \dfrac{(\phi+2p+d+m-1)(1+\epsilon)}{\epsilon(1-\epsilon)}\right\rceil$
\State $\mathcal{Y}\leftarrow \varnothing$; $\mathcal{X}\leftarrow \{(X_1,\ldots,X_p)|~X_j\subseteq [n], |X_j|\le \lambda,j\in[p]\}$
\For{each $(X_1,\ldots,X_p)\in \mathcal{X}$}
\For{$j\in[p]$}
\State $Y_j\leftarrow \{k\in [n]\backslash X_j|a_{jk}>\min\{a_{jk'}|k'\in X_j\}\}$
\EndFor
\State $\mathcal{Y}\leftarrow \mathcal{Y}\cup (Y_1,\ldots,Y_p)$
\EndFor
\For{each tuple  $(\lambda_1,\ldots, \lambda_{p-1},\eta_1,\ldots, \eta_{p})\in\Lambda\times \Gamma$}
\For{each tuple  $(X_1,\ldots, X_{p},Y_1,\ldots, Y_{p})\in\mathcal{X}\times \mathcal{Y}$}
\If{$\mathcal{U}\cap\mathcal{V}=\varnothing$}
\label{check}
\State Find (if exists) an $\epsilon$-optimal solution $x^*$ to the convex program (CP2$[\cU]$)
\State{Define a polytope $\cP(\cU)\subseteq [0,1]^N$}
\State Find an BFS $y$ of $\cP(\cU)$ such that $a_p^Ty\ge a_p^Tx^*_N$
\State $\overline x\leftarrow \{(\overline x_1,\ldots, \overline x_n)|\overline x_k=\left\lfloor y_k\right\rfloor \text { for }k\in N,~ x_k=1, \text{ for }k\in\cU,~ x_k=0, \text{ for }k\in\cV\}$\Comment{{\em rounding down solution $y$}}
\State $P_{\varepsilon}\leftarrow P_{\varepsilon} \cup \overline x$
\EndIf
\EndFor
\EndFor
\State \Return $P_{\varepsilon}$
\end{algorithmic}
\end{algorithm}

\begin{theorem}
\label{theo:approx-pareto}
Let $I$ be an instance of the BMOQ problem. For every fixed $\epsilon>0$, the Algorithm \ref{mul-PTAS} runs in polynomial time in size of the input and produces an $\epsilon$-approximate Pareto-optimal frontier $P_{\epsilon}$ for $I$.
\end{theorem}
\begin{proof}
It is not difficult to see that the running time of Algorithm \ref{mul-PTAS} is polynomial in size of the input as the parameters $\epsilon,\epsilon',\mu,\lambda$ are constants. Hence, we need only to show the approximation factor of the algorithm. More precisely, let $S$ be an arbitrary feasible solution, we need to prove that there exists a solution $S'\in P_{\epsilon}$ such that $f_j(S')\ge (1-\epsilon)f_j(S)$ and $g_j(S')\le g_j(S)/(1-\epsilon)$, for all $j\in[p]$. Let us define 
\[
\lambda_j =  \max\{\lambda| \alpha_{j\lambda} \le f_j(S)\}
\quad \text{and} \quad
\eta_j =  \min\{\eta| \beta_{j\eta} \ge g_j(S)\}. 
\] 
It follows that 
\[
\alpha_{j\lambda_j} \le f_j(S) \le \alpha_{j\lambda_j}/(1 - \epsilon^2)
\quad \text{and} \quad
\beta_{j\eta_j} \ge g_j(S) \ge (1 - \epsilon^2)\beta_{j\eta_j}
\]

Let $\delta=\min\{\lambda, |S|\}$. For each $j\in[p]$, let $X_j$ be the set that contains $\delta$ items in $S$ with largest utility with respect to the objective $f_j$. One can see that $S$ is a feasible solution to (CP2$[\cU]$) with the tuple $(\lambda_1,\ldots, \lambda_{p-1},\eta_1,\ldots, \eta_{p})$ and the tuple $(X_1,\ldots, X_p)$. Let $y$ be a BFS to the polytope $\cP(\cU)$ and let $\overline x$ be the corresponding rounded solution. By the same arguments provided as in the section \ref{}, one can easily proof that $\overline x$ is feasible. Denote $S'$ by set of items corresponding to the integral solution $\overline x$. We have:
\begin{equation}
\label{eq:g}
g_j(S')=b_j^T\overline x\le \beta_{j\eta_j}\le g_j(S)/(1-\epsilon^2)\le g_j(S)/(1-\epsilon)
\end{equation}
Furthermore, for every $j\in [p-1]$: 
\[
a_j^Tx^* \ge \alpha_{j\lambda_j} \ge (1-\epsilon^2)f_j(S),
\]
and for $j = p$, we even have $a_p^Tx^* \ge (1-\epsilon)\opt \ge (1-\epsilon) f_p(S)$, because $x^*$ is $\epsilon$-optimal to (CP2$[\cU]$). We consider two cases below.

Case I: If $\delta=|S|$, (i.e. $|S|\le \lambda$), the set $S'$ will contain $S$ as a subset. Hence, we have $f_j(S')\ge f_j(S)$ for all $j\in [p]$.

Case II: If $\delta=\lambda$, we prove that $f_j(S')\ge (1-\epsilon)f_j(S)$ for all $j\in[p]$. Let $a_{jk}$ be the smallest utility among the $\lambda$ items with highest utility in $S$ with respect to objective $f_j$, we have 
\[
a_j^Ty+a_j^T\bone_\cU\ge a_j^Tx^*_N+a_j^T\bone_\cU=a_j^Tx^*\ge (1-\epsilon)\OPT \ge (1-\epsilon)\lambda a_{jk}
\]
 On the other hand, since $\overline x$ is obtained by rounding down $y$, and the fact that $y$ has at most $\phi+2p+d+m-1$ fractional components, we have 
$$f_j(S')=a_j^T\overline x =a_j^T\overline x_N+a_j^T\bone_\cU\ge a_j^Ty- a_{jk}(\phi +2p+d+m-1)+a_j^T\bone_\cU$$
Therefore, 
\begin{align*}
f_j(S')\ge  a_j^T x^*-\dfrac{\phi+2p+d+m-1}{\lambda}\lambda a_{jk}\ge & \,\,a_j^Tx^* - \dfrac{\phi+2p+d+m-1}{\lambda(1-\epsilon)}a_j^Tx^*\\
= &\,\,(1-\dfrac{\phi+2p+d+m-1}{\lambda(1-\epsilon)})a_j^T x^*\\
\ge &\,\, (1-\dfrac{\epsilon}{1+\epsilon})a_j^T x^*=\dfrac{a_j^T x^*}{1+\epsilon}
\end{align*}
Note that $a_j^T x^* \ge (1-\epsilon^2)f_j(S)$, we have 
\begin{equation}
\label{eq:f}
f_j(S')\ge \dfrac{a_j^T x^*}{1+\epsilon}\ge \dfrac{1-\epsilon^2}{1+\epsilon}f_j(S)= (1-\epsilon)f_j(S).
\end{equation}
From (\ref{eq:g}) and (\ref{eq:f}), the proof is completed.\end{proof}

\subsection{Binary Quadratically Constrained Quadratic Programming Problems}
We consider the binary quadratically constrained quadratic programming problem which takes the form:
\begin{eqnarray*}
\textsc{(BQCQP)} \qquad & \displaystyle \max\,\, w(x)=x^TQx+q^Tx\label{}\\
\text{subject to}\qquad & \displaystyle x \in X \subseteq\{0,1\}^n,
\end{eqnarray*}
where $Q\in\mathcal{CP}_n^*$, $q\in \mathbb{R}_+^n$, and the feasible region $X$ is defined as in the previous section.


\begin{theorem}
The problem (BQCQP) admits a PTAS.
\end{theorem}
\begin{proof}
Let $I$ be an instance of (BQCQP) and let $\epsilon>0$ be a fixed constant.  Since $p=\text{cp-rank}(Q)$ is fixed, one can find in polynomial time a factorization $Q=\sum_{i=1}^r{a_ja_j^T}$, $a_j\in \RR_+^{n}$. Hence, the objective function can be written in the form $\sum_{j=1}^p(a_j^Tx)^2+a_{p+1}^Tx$, where $a_{p+1}=q$.  Applying the Theorem \ref{theo:approx-pareto} for the multi-objective instance with $|\mathcal{I}|=p+1,|\mathcal{J}|=0$, we obtain $P_{\epsilon}$ as an approximation of the Pareto-optimal frontier. Let $x^*\in X$ be an optimal solution to the instance $I$, there exists a solution $y\in P_{\epsilon}$ such that $a_j^Ty\ge (1-\epsilon)a_j^Tx^*$, for all $j\in[p+1]$. Hence, 
\[
w(y)=\sum_{j=1}^p(a_j^Ty)^2+a_{p+1}^Ty\ge (1-\epsilon)^2\sum_{j=1}^p(a_j^T x^*)^2+(1-\epsilon)a_{p+1}^Tx^*\ge (1-\epsilon)^2\OPT.
\]
Let $y^* =\argmax\{w(x)|~x\in P_{\epsilon}$\} (note that the size of $P_{\epsilon}$ is polynomial in size of the input), we have $w(y^*)\ge w(y)\ge (1-\epsilon)^2\OPT$.\end{proof}

\subsection{Binary Multiplicative Programming Problems}
We consider the following problem where the objective is the product of a number of linear functions.
\begin{eqnarray*}
\textsc{(BMPP)} \qquad & \displaystyle \max\,\, w(x)= (a_1^Tx)\cdot (a_2^Tx)\cdot \ldots \cdot (a_p^Tx)\label{}\\
\text{subject to}\qquad & \displaystyle x \in X \subseteq\{0,1\}^n,
\end{eqnarray*}
where $a_j\in\RR_+^{n}$ for all $j\in [p]$.
\begin{theorem}
If $p$ is constant, the problem (BMPP) admits a PTAS.
\end{theorem}
\begin{proof}
Again, applying the Theorem \ref{theo:approx-pareto} for the multi-objective instance with $|\mathcal{I}|=p,|\mathcal{J}|=0$, we obtain an approximation of the Pareto set $P_{\epsilon}$. Let $x^*\in X$ be an optimal solution to the instance of (BMPP), there exists a solution $y\in P_{\epsilon}$ such that $a_j^Ty\ge (1-\epsilon)a_j^Tx^*$, for all $j\in[p]$. Hence, 
\[
w(y)=\prod_{j=1}^pa_j^Ty\ge (1-\epsilon)^p\prod_{j=1}^p(a_j^T x^*)=(1-\epsilon)^p\OPT.
\]
Again, let $y^* =\argmax\{w(x)|~x\in P_{\epsilon}$\}, it follows $w(y^*)\ge (1-\epsilon)^p\OPT$. Note that $p$ is constant, the result follows.\end{proof}

 \subsection{Sum of Ratios}
We consider the binary quadratically constrained programming with a rational objective function. 
\begin{eqnarray*}
\textsc{(Sum-Ratio)} \qquad & \displaystyle \max w(x)= {\sum}_{j=1}^p\dfrac{a_j^Tx}{b_j^Tx}\label{}\\
\text{subject to}\qquad & \displaystyle x \in X \subseteq\{0,1\}^n.
\end{eqnarray*}

This problem belongs to the class of fractional programming which has important applications in several areas such as transportation, finance, engineering, statistics (see the survey paper by \cite{sch-shi:j:fractional-programming-the-sum-of-ratio} and the references therein, for more applications).
\begin{theorem}
Suppose that $p$ is fixed. The problem \textsc{(Sum-Ratio)}  admits a PTAS.
\end{theorem}
\begin{proof}
applying the Theorem \ref{theo:approx-pareto} for the instance with $|\mathcal{I}|=|\mathcal{J}|=p$, we obtain an approximation of the Pareto set $P_{\epsilon}$ of the instance of multi-objective problem. Let $x^*\in X$ be an optimal solution to an instance of \textsc{(Sum-Ratio)} problem, there exists a solution $y\in P_{\epsilon}$ such that $a_j^Ty\ge (1-\epsilon)a_j^Tx^*$ and $b_j^Ty\le b_j^Tx^*/(1-\epsilon)$, for all $j\in[p]$. Hence, 
\[
w(y)=\sum_{j=1}^p\dfrac{a_j^Ty}{b_j^Ty}\ge (1-\epsilon)^2\sum_{j=1}^p\dfrac{a_j^Tx^*}{b_j^Tx^*}= (1-\epsilon)^2\OPT.
\]
Again, let $y^* =\argmax\{w(x)|~x\in P_{\epsilon}$\}, the proof is completed.
\end{proof}
